\documentclass[a4paper,12pt]{article}
\usepackage[english]{babel}
\usepackage{graphicx}
\usepackage{authblk}
 \usepackage{amsmath,amsfonts,amssymb,amsthm,epic,eepic,epsf,mathtools,accents,appendix,epsfig,subfigure,float,xcolor}
\usepackage[all]{xy}%
\usepackage{xcolor}
\usepackage[overload]{empheq}

\newcommand\norm[1]{\left\lVert#1\right\rVert}
\newcommand{\mycomment}[1]{}
\newcommand{\ep}{\epsilon}

\newtheorem{lemma}{Lemma}
\newtheorem{theorem}{Theorem}
\newtheorem{proposition}{Proposition}
\newtheorem{corollary}{Corollary}
\newtheorem{remark}{Remark}
\newtheorem{definition}{Definition}

\begin{document}

\title{Asymptotic spreading of predator-prey populations in a shifting environment}

\author[1]{King-Yeung Lam}
\author[1]{Ray Lee}
\affil[1]{\small Department of Mathematics, The Ohio State University, Columbus, OH 43210, USA}

\maketitle

\section*{Abstract}
{\small 
Inspired by recent studies associating shifting temperature conditions with changes in the efficiency of predator  species in {converting their prey to offspring}, we propose a predator-prey model of reaction-diffusion type to analyze the consequence of such effects on the population dynamics and spread of species.
In the model, the predator conversion efficiency is represented by a spatially heterogeneous function depending on the variable $\xi=x-c_1t$ for some given $c_1>0$. Using the Hamilton-Jacobi approach, we provide explicit formulas for the spreading speed of the predator species. When the conversion function is monotone increasing, the spreading speed is determined in all cases and non-local pulling is possible. When the function is monotone decreasing, we provide formulas for the spreading speed when the rate of shift of the conversion function is sufficiently fast or slow.}

\section{Introduction}

A fundamental {challenge} in ecology is to understand the persistence and spread of a given species in an environment.
These {issues} are receiving a renewed interest as changes in climatic conditions can dramatically impact the suitability of a habitat for a species' survival and growth. As temperatures rise, many species have moved in the directions of the poles or toward higher elevations, in an apparent attempt to keep pace with shifting temperature isotherms \cite{parmesan2003,lenoir2020species}. {As species establish in new regions, new biotic interactions take place \cite{gilman2010framework}, which in turn can have significant consequences for species abundance and biodiversity 
\cite{wallingford2020adjusting,ling2008range,sorte2010}, the functioning of ecosystems \cite{sorte2010,verges2016long,weiskopf2020climate}, the spread of disease \cite{pecl2017,wu2016impact} and human welfare \cite{pecl2017}. The ecological effects of a changing climate are complex and various. While many species are vulnerable to a changing climate, for many others climate-related changes may facilitate expansion to new areas and population growth \cite{walther2009alien,aronson2015no,bradley2010climate,hobbs2000invasive}.} 


Mathematical modeling can be used to determine why certain species decline while others prosper under the changing climate. The study of species persistence and spread often depend on the spatial context, and much analysis in the classical literature has been based on reaction-diffusion models. A prominent example is the Fisher-KPP equation \cite{fisher1937wave,kolmogorov1937etude}, which describes the spreading of a single population
\begin{align}\label{eq:Fisher}
    u_t = du_{xx} + ru(1-u) \quad \text{ for } x\in \mathbb{R}, ~ t>0,
\end{align}
where $d$ corresponds to the dispersal rate and $r$ the intrinsic growth rate of a species $u$. 
For a given population density $u(t,x)$, 
Aronson and Weinberger \cite{aronson1975nonlinear,aronson1978multidimensional} introduced the key notion of spreading speed, which refers to the number $c^*>0$ such that
\begin{align*}
    \lim_{t\to\infty} \sup_{|x|>ct} u(t,x) &= 0\quad \text{ for }c\in (c^*,\infty), \quad \text{and} \quad
    \lim_{t\to\infty} \inf_{|x|>ct} u(t,x) &> 0\quad \text{ for }c\in (0,c^*).
\end{align*}

The problem of spreading speed for more general equations of the form
\begin{align}
    u_t = du_{xx} + f(u) \quad \text{ for } x\in \mathbb{R}, ~ t>0,
\end{align}
was first investigated by Kolmogorov et al. \cite{kolmogorov1937etude} for heaviside initial condition, who showed under certain assumptions on the growth function $f(u)$ that 
$$
c^* = 2\sqrt{df'(0)}.$$ This result was later generalized to any compactly supported initial data and in higher spatial dimensions  by Aronson and Weinberger \cite{aronson1978multidimensional}. This theoretical spreading speed has yielded good estimates for range expansion observed in nature
\cite{shigesada1997biological}.


Various studies have since revisited Fisher's model with an interest in the impact of a shifting environment. To consider climate change, it is often assumed that the behavior of the species depends on the variable $\xi = x-ct$, where the constant $c$ corresponds to the velocity of a shifting climate \cite{potapov2004climate,berestycki2009can,li2014persistence,berestycki2018forced}. See also \cite{fang2016can}, which studied a similar model in the context of an SIS model, and \cite{wang2022recent}, for a survey on reaction-diffusion models in shifting environments. Many models have proposed the case where the growth rate $ru(1-u)$ in \eqref{eq:Fisher} is replaced by a shifting logistic form $u(r(x-ct)-u)$, where $r(x-ct)$ denotes the species' intrinsic growth rate \cite{potapov2004climate,berestycki2009can}. These works assumed the growth rate $r(\xi)$ to be positive on a bounded patch of suitable habitat and negative elsewhere, and were broadly interested in the effects of a shifting climate on the persistence of the species.

A shifting environment also leads to new spreading phenomena. In \cite{lam2022asymptotic}, the spreading speed for solutions of Fisher's equation with growth rate $f(u,x-ct) = (r(x-ct) - u)u$ was determined using the Hamilton-Jacobi method, in the case that the intrinsic growth rate of the species is positive and monotone. 
%
%
They showed that, for a certain range of velocities of climate shift, the species spreads with speed distinct from either of the limiting KPP invasion speeds in a phenomenon called non-local pulling \cite{holzer2014accelerated,girardin2019invasion}. {When the growth rate is non-monotone, the existence of forced waves and their attractivity is studied in \cite{berestycki2018forced}.}


In addition to single species equations, the spreading dynamics for systems of equations has received considerable attention. Building on the earlier works on order-preserving systems (such as cooperative systems and competitive systems of two species)   \cite{lewis2002spreading,weinberger2002analysis,li2005spreading,liang2007asymptotic}, the spreading of two competing species in a shifting habitat is studied in \cite{Zhang2017persistence,Dong2021forced}.

By contrast, for predator-prey systems a comparison principle is not immediately available and many studies regarding propagation phenomena in these systems have focused on the dynamics of traveling wave solutions. The existence of traveling wave solutions for two-species predator prey equations was established in \cite{dunbar1983travelling,gardner1984existence}, and studied further in \cite{huang2003existence,mischaikow1993travelling}, while some results on the stability of traveling wave solutions were established in \cite{gardner1991stability}. 

Until recently, few works have treated the spreading dynamics of predator-prey systems with general initial data. In \cite{pan2017invasion}, Pan determined the spreading speed of the predator for a predator-prey system with initially constant prey density and compactly-supported predator. Shortly thereafter, Ducrot, Giletti, and Matano \cite{ducrot2019spreading} used methods from uniform persistence theory to characterize the spreading dynamics when both predator and prey are initially compactly supported. They showed that the behavior can be classified based on the speeds of the prey in the absence of predator, and of the predator when prey is abundant (see also \cite{ducrot2019thespreading}, for the case of a predator-prey equation with non-local dispersal, and \cite{ducrot2013convergence,pan2013asymptotic}). Since these works, the spreading speeds regarding the Cauchy problem for predator-prey systems with three species was studied in \cite{Ducrot2021asymptotic} (see also \cite{wu2019spreading}). There, it was shown that the nonlocal pulling phenomenon can occur in a system with two predators and one prey.

For other types of non-cooperative systems and their spreading speeds, we refer to \cite{Wang2011spreading}, which characterizes the spreading speed for a general class of non-cooperative reaction diffusion systems as the minimal traveling wave speed. We also note \cite{ducrot2016spatial}, which determined the spreading speed of infectious disease in an epidemic model, and \cite{Liu2021stacked}, which considered the spreading dynamics for competition between three species.  Spreading dynamics are also studied for nonlocal diffusion problems, here we mention \cite{Li2021invasion,Wu2023propagation} for such results in predator-prey models in the absence of shifting environment.

\begin{subsection}{The predator-prey model in a shifting environment}

We are interested in the effect of the heterogeneous shifting profile of the conversion efficiency of prey to predator, represented below by the function  ${{\tilde a}}(x-c_1t)$, on the spreading dynamics. Though the temperature-dependence of the conversion efficiency is not well-understood, there is some evidence that the conversion efficiency is impacted by climate. Using an experimental system of predator and prey, Daugaard et al. \cite{daugaard2019} found that the conversion efficiency of the predator increased with warming, and in a recent meta-analysis Lang et. al \cite{lang2017temperature} identified a trend toward increasing efficiency of energy assimilation {by consumers} with increasing temperature. {On the other hand, many biological processes depend unimodally on temperature, such that measures of species performance and fitness decline once temperature increases sufficiently beyond a ``thermal optimum" \cite{huey1979integrating,deutsch2008impacts,buckley2012functional}. It is thus plausible that predators currently experiencing climates at or near their thermal optimum may experience declines in conversion efficiency with additional warming.}

To this end, we propose the following predator-prey model of reaction-diffusion type to analyze the consequence of such effects on the population dynamics and spread of species:
\begin{equation}\label{eq:dim_system}
\begin{cases}
\tilde{u}_t= d_1 \tilde{u}_{xx} + \tilde{u}(-\kappa - \alpha_1 \tilde{u} + \tilde{a}(x-c_1 t) \tilde v)  &\text{ in }(0,\infty)\times\mathbb{R},\\
\tilde{v}_t= d_2 \tilde{v}_{xx} + \tilde{v}(\tilde r - \alpha_2 \tilde{v} - \tilde{b} \tilde u)  &\text{ in }(0,\infty)\times\mathbb{R},\\
\tilde{u}(0,x) = \tilde{u}_0(x),\quad \tilde{v}(0,x)=\tilde{v}_0(x) &\text{ in }   \mathbb{R}. 
\end{cases}
\end{equation}
Here the predator and prey densities are represented by
$\tilde{u}(t,x)$ and $\tilde{v}(t,x)$. 
It is assumed that the predator cannot persist in the absence of prey, and competes with other predators, while in the absence of predation
the prey exhibits logistic growth and is described by the standard Fisher-KPP equation. The interaction rates between predator and prey are mediated by the consumption rate $\tilde{b} > 0$ of prey by predator, and by the predator's conversion efficiency function $\tilde{a}(x-c_1t)$, which describes the degree to which consumed prey can be successfully converted to additional predators. For simplicity, we assume that conversion efficiency has a fixed profile  in the moving coordinate $y= x-c_1 t$ with constant velocity $c_1$. 
Finally, $\tilde d_i, \tilde \alpha_i, \kappa, \tilde r$ are positive parameters, where $\tilde{d}_i$ are the random dispersal rates, $\alpha_i$ are the intraspecific competition rates, $\kappa$ is the natural death rate of the predator species and $\tilde r$ is the natural birth rate of the prey species.

Without loss of generality, we may non-dimensionalize the problem \eqref{eq:dim_system} and obtain the following model:
\begin{align}\label{eq:system}
    \begin{cases}
        u_t = u_{xx} +  \big(-1- u + a(x-c_1t)v\big)u  \quad &\text{in } (0,\infty) \times \mathbb{R}\\
        v_t = dv_{xx}  + r(1 - v - bu)v \quad &\text{in } (0,\infty) \times \mathbb{R}\\
        u(0,x)=u_0(x), \quad v(0,x)=
        v_0(x) &\text{ in }\mathbb{R}.
    \end{cases}
\end{align}

We assume the following throughout our study of \eqref{eq:system}. 
\begin{itemize}

    \item[{\bf(H1)}] The function $a:\mathbb{R} \to \mathbb{R}$ is  monotone, and  satisfies
    $$
        \beta := 1-b(\norm{a}_\infty - 1) > 0,\qquad 
        \inf_{s \in \mathbb R} a(s) > \tfrac{1}{\beta}, 
        \quad \text{ and }\quad 
        \norm{a}_\infty > 1.
   $$
\end{itemize}

By observing (via maximum principle) that the density of the predator is bounded from above by $\|a\|_\infty -1$, it follows that the quantity $\beta \coloneqq 1-b(\norm{a}_\infty - 1)$ corresponds to the minimum carrying capacity for the prey.

\begin{remark}
It is documented   in a microbial predator-prey system \cite{daugaard2019} that the quantity of predators produced for each prey consumed increases when temperature is increased. This corresponds to the case when $c_1>0$ and $a(\cdot)$ is decreasing. We also study the case when increasing temperature decreases the predator efficiency, i.e. $a(\cdot)$ is increasing.  
\end{remark}


We are interested in the situation when the initial data of the predator is compactly supported, while that of the prey has a positive upper and lower bound. For simplicity, we will assume throughout the discussion that $(u_0,v_0) \in C^2(\mathbb R)$ satisfies
\begin{itemize}
    \item[{\bf (IC)}]  
        $0 \leq u_0 \leq \norm{a}_\infty-1,~ \beta \leq v_0 \leq 1$, 
    and
        $u_0$ \text{has compact support}.
\end{itemize}
\end{subsection}

Finally, we define the following limiting growth rates (at $\pm \infty$), to simplify the statements and proofs of the main results.
\begin{equation}\label{eq:param}
\begin{cases}
 r_1 = a(-\infty) - 1, \qquad &r_2 = a(+\infty) - 1,\\
    \underline r_1 = \beta a(-\infty) - 1, \qquad    
   &\underline r_2 = \beta a(+\infty) - 1. 
    \end{cases}
\end{equation}  
Here, $\underline r_1$ and $\underline r_2$ correspond to the limiting growth rates of the predator behind and ahead of the environmental shift, respectively, when the prey density is at its minimum value $v=\beta$, while $r_1$ and $r_2$ are the limiting growth rates of the predator behind and ahead of the shift, respectively, when the prey density is at its maximum value $v=1$.

\subsection{Main Results}

In this paper, we are interested in the asymptotic speed of spread (or spreading speed) as the predator species $u$ expands its territory.  Up to a change of coordinates $x \mapsto -x$, it is enough to focus our discussion on the rightward spreading speed, while allowing the spatial heterogeneity ${a}(\cdot)$ to be monotonically increasing or decreasing.

In the remainder of this paper, we will refer to the rightward spreading speed $c^*$ simply as the spreading speed, which is defined as follows. 
\begin{definition}
Let $u$ be the solution of \eqref{eq:system}, where $u_0$ and $v_0$ satisfy (IC).
We say that the species $u$ has spreading speed given by $c^*>0$ if
\begin{equation}
\lim_{t\to\infty} \sup_{x > c t} u(t,x) = 0 \quad \text{ for each }c \in (c^*,\infty), 
\end{equation}
\begin{equation}
\lim_{t\to\infty} \inf_{0<x < c t} u(t,x) > 0 \quad \text{ for each }c  \in (0,c^*). 
\end{equation}
\end{definition}
Following \cite[Definition 1.2]{Hamel2012spreading}, we introduce the closely related notion of maximal and minimal speed $\overline c_*$, $\underline c_*$:
\begin{definition}
Let $u$ be the solution of \eqref{eq:system}, where $u_0$ and $v_0$ satisfy (IC).
\begin{align}\label{eq:max/min_spd}
\begin{cases}
    \overline c_* &\coloneqq \inf\left\{c > 0 \mid \limsup\limits_{t \to \infty}\sup\limits_{x > ct} u(t,x) = 0 \right\},\\
    \underline c_* &\coloneqq \sup\left\{c > 0 \mid \liminf\limits_{t \to \infty} \inf\limits_{0 \leq x < ct} u(t,x) > 0 \right\}. 
\end{cases}
\end{align}
\end{definition}
\begin{remark}
The species $u$ has a spreading speed if and only if $\overline c_* = \underline c_*$. In such a case, the spreading speed $c^*$ is given by the common value $\overline c_* = \underline c_*$.
\end{remark}



The following two main theorems characterize the spreading speed of $u$ for the cases (i) $a(\cdot)$ is monotonically increasing and (ii) $a(\cdot)$ is monotonically decreasing, respectively. Assuming the positive axis {points} poleward and temperature is {rising}, they correspond to the cases when the the conversion efficiency of the predator is suppressed or enhanced by the warming climate.

\begin{theorem}\label{thm:increasing}
Let $c_1>0$ be given,   $a:\mathbb{R}\to \mathbb{R}$ be increasing, and suppose {\rm{\bf(H1)}} holds.  If $(u(t,x),v(t,x))$ is the solution of \eqref{eq:system} with initial data satisfying {\rm{\bf(IC)}}, then the spreading speed of $u$ exists, and is given by 
    \begin{align}\label{eq:sigma1(c_1)}
        c^* \coloneqq \begin{cases}
        2\sqrt{r_2} & \text{ if } c_1 \leq 2\sqrt{r_2}\\
        \frac{c_1}{2} - \sqrt{r_2-r_1} + \frac{r_1}{\frac{c_1}{2} - \sqrt{r_2-r_1}} & \text{ if } 2\sqrt{r_2} < c_1 < 2\sqrt{r_1} + 2\sqrt{r_2-r_1}\\
        2\sqrt{r_1} & \text{ if } c_1 \geq 2\sqrt{r_1} + 2\sqrt{r_2-r_1}.
        \end{cases}
    \end{align}
\end{theorem}

\begin{theorem}\label{thm:decreasing}
Let $c_1 \in \mathbb{R} \setminus (2\sqrt{r_2},2\sqrt{r_1}]$ be given, $a:\mathbb{R}\to \mathbb{R}$ be decreasing, and suppose {\rm{\bf(H1)}} holds. If $(u(t,x),v(t,x))$ is the solution of \eqref{eq:system} with initial data satisfying {\rm{\bf(IC)}}, then the spreading speed of $u$ exists, and is given by
$$
c^* = \begin{cases}
    2\sqrt{r_2} &\text{ if }c_1 \leq 2\sqrt{r_2},\\
    2\sqrt{r_1} &\text{ if }c_1 > 2\sqrt{r_1}.
\end{cases}
$$

\end{theorem}

For the proof of Theorems \ref{thm:increasing} and \ref{thm:decreasing}, see Subsection \ref{subsec:5.2}

We also show  the convergence of $(u,v)$ to the homogeneous state in the moving frames with speed different from $c_1$ or $c^*$. 
\begin{theorem}\label{thm:main3}
    For $c \in (0,\infty)\setminus\{c_1,c^*\}$, let 
$$
(u_c(t,x),v_c(t,x)) = (u(t,x+ct),v(t,x+ct)).
$$
\begin{itemize}
\item[{\rm(a)}] If $c \in (c^*,\infty)$, then $\displaystyle \lim_{t\to\infty}(u_c(\cdot,t),v_c(\cdot,t)) \to (0,1)$ in $C_{loc}({\mathbb{R}})$.
\item[{\rm(b)}] If $c \in (0,c^*) \setminus \{c_1\}$, then 
$$
\lim_{t\to\infty}(u_c(\cdot,t), v_c(\cdot,t)) \to (\tilde{u},\tilde{v})\quad \text{ in }C_{loc}({\mathbb{R}}),
$$
where $(\tilde u, \tilde v)$ is the unique positive equilibrium of the kinetic system
$$
\frac{d}{dt}\tilde u= 
\tilde u(-1-\tilde u + A \tilde v),\qquad \frac{d}{dt} \tilde v= r\tilde v(1-\tilde v - b\tilde u),
$$
such that $A = a(-\infty)$ if $c<c_1$ and $A = a(+\infty)$ if $c> c_1$.    
\end{itemize}
\end{theorem}
For the proof of Theorem \ref{thm:main3}, see Section \ref{sec:homo}.

\begin{remark}
Note that the case $c_1 \in (2\sqrt{r_2}, 2\sqrt{r_1}]$ is not covered by Theorem \ref{thm:decreasing}. In that case {the} Hamilton-Jacobi approach does not directly apply. We conjecture that $c^* = c_1$ in that case and the predator advances in locked step with the environment. See \cite{berestycki2018forced,fang2016can} for results regarding a single species in a shifting habitat. 
A possible approach is to use the persistence theory as in \cite{Ducrot2021asymptotic}.
\end{remark}

\begin{remark}
In the case of $a(\cdot) \equiv a_0$ being a constant {and $v_0 \equiv 1$}, the spreading speed of the predator was determined by Pan in \cite[Theorem 2.1]{pan2017invasion}.
\end{remark} 

To consolidate the formulas for the spreading speeds in Theorems \ref{thm:increasing} and \ref{thm:decreasing}, we will denote $\lambda^*= \frac{c_1}{2} - \sqrt{|r_2-r_1|}$. Then
the spreading speed for all cases can be given by
\begin{align}\label{eq:sss}
   \sigma(c_1;r_1,r_2) = \begin{cases}
        2\sqrt{r_2} & \text{ if }   r_1 < r_2 \quad \text{and}\quad c_1 \leq 2\sqrt{r_2}\\
        \lambda^* + \frac{r_1}{\lambda^*} & \text{ if }  r_1 < r_2 \quad \text{and}\quad 2\sqrt{r_2} < c_1 < 2\sqrt{r_1} + 2\sqrt{r_2-r_1}\\
        2\sqrt{r_1} & \text{ if }  r_1 < r_2 \quad \text{and}\quad c_1 \geq 2\sqrt{r_1} + 2\sqrt{r_2-r_1},\\
        2\sqrt{r_2} & \text{ if } r_1 > r_2 \quad \text{and} \quad c_1 \leq 2\sqrt{r_2}\\
        2\sqrt{r_1} & \text{ if } r_1 > r_2 \quad \text{and} \quad c_1 > 2\sqrt{r_1},
    \end{cases}
\end{align} 
or, equivalently,
\begin{align}\label{eq:sigma(c_1)}
    \sigma(c_1;r_1,r_2) = \begin{cases}
        2\sqrt{r_2} & \text{ if }  c_1 \leq 2\sqrt{r_2}\\
        \lambda^* + \frac{r_1}{\lambda^*} & \text{ if }  r_1 < r_2 \quad \text{and}\quad 2\sqrt{r_2} < c_1 \leq 2\sqrt{r_1} + 2\sqrt{r_2-r_1}\\
        2\sqrt{r_1} & \text{ if }   c_1 > 2\sqrt{r_1} + 2\sqrt{\max\{0,r_2-r_1\}},
    \end{cases}
\end{align}

\subsection{Related mathematical results}

We also mention a closely related work of Choi, Giletti, and Guo \cite{choi2021persistence}, where they considered a two-species predator-prey system similar to \eqref{eq:system}, with the intrinsic growth rate $r=r(x-c_1t)$ for the prey subject to the climate shift instead of the coefficient $a$.
They considered the case when both initial data $u_0$ and $v_0$ are compactly supported and a non-decreasing profile for the growth rate with $r$ changing sign, $r(-\infty) < 0 < r(\infty)$. In the case of local dispersal, they showed that the prey persists by spreading if and only if the maximal speed of the prey exceeds the environmental speed (i.e., $2\sqrt{d r (+\infty)} > c_1$), while the predator persists by spreading at the speed given by the smaller of the prey and maximal predator spreading speeds. In their setting both species tends to zero in $\{(t,x): x< c_1t\}$, while in the zone ahead of the environmental shift, the density of the prey is strictly decreasing so there is no nonlocal pulling phenomenon.
We also mention \cite{guo2023spreading} for the case of two weak-competing predators and one prey, and \cite{ahn2022spreading}, for the case of one predator and two preys. For compactly supported initial data, the invasion wave of the prey resembles the effect of a shifting environment studied in our paper. However, the exact spreading speed of the predator(s) is not completely determined. 

Finally, we mention the work of Bramson \cite{bramson1983convergence}, which established using probabilistic techniques a correction term of $\frac{3}{2}\log t$ which separates the the location of the spreading front for solutions to the Fisher-KPP equation \eqref{eq:Fisher} and the asymptotic location of the minimal traveling wave solution. This result was later generalized using maximum principle arguments by Lau to KPP-like nonlinearities $f(s)$ satisfying $f'(s) \leq f'(0)$ on $[0,1]$ \cite{lau1985on}. For systems of equations of predator-prey type, the existence and characterization of such a delay between the spreading front and the asymptotic rate of spread is a challenging open question.

\subsection{Organization of the paper}
The rest of the paper will be organized as follows. In Section \ref{sec:2}, we give a quick proof of the upper estimate of the spreading speed (namely, $\overline{c}_* \leq \sigma(c_1;r_1,r_2)$) by invoking the recent results on the diffusive logistic equation in shifting environment due to \cite{lam2022asymptotic}. In Section \ref{sec:3}, we derive some rough estimates for the prey density $v(t,x)$, and state five separate cases for the key parameters $c_1,r_1,r_2$ where the spreading speed has to be treated separately. In Section \ref{sec:4}, we outline, in several lemmas, the conceptual steps to estimate the spreading speed from below via explicit solution of some Hamilton-Jacobi equation \eqref{eq:rho_hj} obtained as the limiting problem of the first equation of \eqref{eq:system}. These lemmas will be proved in Subsections \ref{4.1}, \ref{subsection:4.2} and \ref{section:rho}. In Section \ref{sec:explicitform}, we determine the explicit formulas of the unique solution $\hat\rho$ to the limiting problem in each case, and prove that the upper bound of $c^*$  obtained in Section \ref{sec:3} is also the lower bound. This finishes the proofs of
{Theorems \ref{thm:increasing} and \ref{thm:decreasing}} regarding the spreading speed. In Section \ref{sec:homo}, we prove Theorem \ref{theorem:steadystate} regarding the convergence to homogeneous state. Finally, in the Appendix, we collect some useful comparison results regarding the limiting Hamilton-Jacobi equations in \cite{lam2022asymptotic}, which are rephrased in a format suitable for our purpose here.

\section{Upper bound on spreading speed}\label{sec:2}
In this section, we give a quick proof of $\overline{c}_* \leq \sigma(c_1;r_1,r_2)$, where $\sigma(c_1;r_1,r_2)$ is given by the first three cases of \eqref{eq:sss},
i.e., the spreading speed $c^*$ is bounded above by $\sigma(c_1;r_1,r_2)$. 

First, we establish some preliminary estimates on the solutions $u$ and $v$ of \eqref{eq:system}.

\begin{lemma}\label{lemma:gl_bds}
Assume $0 \leq u_0 \leq \norm{a}_\infty-1$ and $\beta \leq v_0 \leq 1$. Then the corresponding solutions $u(t,x)$ and $v(t,x)$ of \eqref{eq:system} satisfy $0 \leq u(t,x) \leq \norm{a}_\infty-1$ and $\beta \leq v(t,x) \leq 1$ for all $(t,x) \in (0,\infty) \times \mathbb{R}$.
\end{lemma}
\begin{proof}
    By the classical theory of reaction-diffusion equations, there exists a unique solution $(u,v)$ satisfying \eqref{eq:system} for all $(t,x) \in (0,\infty) \times \mathbb R$; see, e.g. \cite{Smoller1983shock}. Moreover, since $0 \leq u_0 \leq \norm{a}_\infty-1$ and $\beta \leq v_0 \leq 1$, the maximum principle (see \cite[Chapter 3, Section 6, Theorem 10]{Protter1984maximum} or  \cite[Theorem 6.2.1]{lam2022introduction}) implies $0 \leq u \leq \norm{a}_\infty - 1$ and $\beta \leq v \leq 1$ on $(0,\infty) \times \mathbb R$.
\end{proof}

The global upper bound for $v$ established in Lemma \ref{lemma:gl_bds}, combining with existing results for the diffusive logistic equations with heterogeneous shifting coefficients \cite{lam2022asymptotic}, can be used to determine an upper bound for the spreading speed of $u$.
\begin{proposition}\label{lemma:upper_speed}
   Let $(u(t,x),v(t,x))$ be the solution of \eqref{eq:system}, with the associated maximal spreading speed $\bar{c}_*$ as given in \eqref{eq:max/min_spd}. Then 
    \begin{align}\label{eq:speed_upper_bd}
        \lim_{t \to \infty} \sup_{x \geq (\sigma+\eta)t} u(t,x) = 0 \quad \text{ for each }\eta>0.
    \end{align}
In particular, $\overline{c}_* \leq \sigma$, where $\sigma = \sigma(c_1;r_1,r_2)$ is defined in \eqref{eq:sigma(c_1)}.
\end{proposition}
\begin{proof}
By Lemma \ref{lemma:gl_bds}, $v(t,x) \leq 1$ for all $(t,x) \in (0,\infty)\times\mathbb{R}$, hence we may regard $u(t,x)$ as a subsolution of the following scalar problem
    \begin{equation}\label{eq:super}
        \begin{cases}
             \bar{u}_t = \bar{u}_{xx} + \bar u(-1 - \bar u + a(x-c_1t)) \quad \text{in } (0,\infty) \times \mathbb R\\
             \bar{u}(0,x) = u_0(x) \quad \text{in } \mathbb R.
        \end{cases}
    \end{equation}

Let $\bar{u}$ be the classical solution of \eqref{eq:super} with initial data $u_0(x)$.   By the parabolic maximum principle we have 
  \begin{equation}\label{eq:compare11}
  u(t,x) \leq \bar u(t,x)\quad \text{ for all }~(t,x) \in (0,\infty) \times \mathbb R.
  \end{equation}

In the case $r_2 > r_1$, we may invoke  \cite[Theorem 6]{lam2022asymptotic}  to deduce that $\bar u$ satisfies
    \begin{equation}\label{eq:SSS}
        \lim_{t \to \infty} \sup_{x \geq (\sigma+\eta)t} \bar u(t,x) = 0 \quad \text{ for each }\eta > 0,
    \end{equation}
    where $\sigma = \sigma(c_1;r_1,r_2)$ is given in \eqref{eq:sigma(c_1)}.

In case $r_1 > r_2$, we define 
$$
U(t,x) = \begin{cases}
    \exp(-\sqrt{r_1}{x} + 2r_1 t) &\text{ if } c_1 \geq 2\sqrt{r_1},\\
    \min\{r_1,\exp(-\lambda(x-c_1 t)) \}&\text{ if }2\sqrt{r_2} < c_1 < 2\sqrt{r_1},~ \lambda=\tfrac{1}{2}\left( c_1-\sqrt{c_1^2 - 4r_2}\right)\\
    \min\{r_1, \exp(- \sqrt{r_2}x + 2r_2 t)\} &\text{ if } 0<c_1 \leq 2\sqrt{r_2}.
 \end{cases}
$$
Then it can be verified that $U$ is a generalized supersolution of \eqref{eq:super} 
(see \cite[Definition 1.1.1]{lam2022introduction} or \cite[Definition~4.2]{Berestycki2016shape} for the definition). Hence, we again deduce that \eqref{eq:SSS} holds where 
$$\sigma = \begin{cases}
    2\sqrt{r_1} &\text{ if }c_1 \geq 2\sqrt{r_1},\\
    2\sqrt{r_2} &\text{ if }c_1 \leq 2\sqrt{r_2},\\    
    c_1 &\text{ otherwise.}
\end{cases}
$$
Combining with \eqref{eq:compare11}, we conclude that
    \begin{align*}
        \lim_{t \to \infty} \sup_{x \geq (\sigma+\eta)t} u(t,x) = 0 \quad \text{ for each }\eta>0,
    \end{align*}
where $\sigma$ is given in \eqref{eq:sigma(c_1)} (and $\sigma = c_1$, in case $2\sqrt{r_1}> c_1 > 2\sqrt{r_2}$).
    This completes the proof.
\end{proof}

\section{Rough Estimate for $v(t,x)$}\label{sec:3}

Having established that the spreading speed is bounded above by $\sigma=\sigma(c_1;r_1,r_2)$, we may also deduce in the following lemma that $v(t,x)$ converges to its carrying capacity as $t \to \infty$ in the region $\{(x,t) : x  >\sigma t\}$.

\begin{lemma}\label{lemma:prey_spd}
Let $(u(t,x),v(t,x))$ be the solution of \eqref{eq:system}. Then
\begin{align}\label{eq:v_ahead}
    \lim_{t \to \infty} \sup_{x \geq (\sigma+ \eta)t} |v(t,x) - 1| = 0 \quad \text{ for each }\eta >0,
\end{align}
where $\sigma$ is given by \eqref{eq:sigma(c_1)}.
\end{lemma}
\begin{proof}
Since $v(t,x) \leq 1$ (thanks to Lemma \ref{lemma:gl_bds}), it suffices to show the lower bound. We shall follow the proof of Theorem 5.1 in \cite{Ducrot2021asymptotic}.  
Fix $c > \sigma(c_1;r_1,r_2)$.
 we may suppose for contradiction that there exists a sequence $\{(t_n,x_n)\}$ with $t_n \to \infty$ and $x_n \geq ct_n$ such that $\limsup\limits_{n \to \infty} v(t_n,x_n) < 1$. Denote $(u_n,v_n)(t,x) = (u,v)(t+t_n,x+x_n)$. By standard parabolic estimates, we may pass to a further subsequence so that $(u_n,v_n)$ converges to an entire in time solution $(u_\infty,v_\infty)$ of \eqref{eq:system} in $C_{loc}(\mathbb{R}^2)$. Since $c > \sigma(c_1;r_1,r_2)$, by Proposition \ref{lemma:upper_speed}, we have $u_\infty \equiv 0$. Thus, $v_\infty$ is an entire  solution satisfying the equation
\begin{align*}
    (v_\infty)_t = d(v_\infty)_{xx} + rv_\infty(1-v_\infty) \quad \text{ for }(t,x) \in \mathbb{R}^2.
\end{align*}
Since $v \geq \beta$  for all $(t,x) \in (0,\infty)\times \mathbb{R}$, we deduce that $v_\infty \geq \beta$ 
for all $(t,x) \in \mathbb{R}^2$. By the classification 
of entire solution of the diffusive logistic equation (see, e.g. \cite[Lemma 2.3(d)]{Liu2020Asymptotic}) we have $v_\infty \equiv 1$. This is in contradiction with the statement $\limsup\limits_{n \to \infty} v(t_n,x_n) < 1$. 
\end{proof}

Having established the upper bound of the spreading speed, the outstanding task is to estimate the spreading speed from below. We will do {so} by adopting the Hamilton-Jacobi approach \cite{Evans1989PDE}. To this end, define
\begin{equation}\label{eq:Fep}
 F^\ep(t,x) = -1 + {v(\tfrac{t}{\ep},\tfrac{x}{\ep}) a\left({\tfrac{x}{\ep}-\tfrac{c_1t}{\ep}}\right)}   
\end{equation}
and its (lower) half-relaxed limit \cite{Barles2013introduction}
\begin{equation}\label{eq:F*}
F_*(t,x) = \liminf_{\ep \to 0 \atop (t',x') \to (t,x)} F^\ep(t',x').
\end{equation}
Thanks to {\bf(H1)}, the function $a(s)$ is monotone.



We will divide the proof of the spreading speed into the following cases, depending on the speed of environmental shift $c_1$ and the profile of the conversion efficiency $a(x-c_1t)$. 
\begin{description}
    \item[Case 1(a):]$r_1< r_2$ and $c_1 \leq 2\sqrt{r_2}$
    \item[Case 1(b):] $r_1 < r_2$ and $2\sqrt{r_2} < c_1 < 2(\sqrt{r_1} + \sqrt{r_2-r_1})$
    \item[Case 1(c):] $r_1< r_2$ and $c_1 \geq 2(\sqrt{r_1} + \sqrt{r_2-r_1})$
    \item[Case 2(a):] $r_1 > r_2$ and $c_1 < 2\sqrt{r_2}$
    \item[Case 2(b):] $r_1 > r_2$ and $c_1 = 2\sqrt{r_2}$
    \item[Case 2(c):] $r_1 > r_2$ and $c_1 > 2\sqrt{r_1}$
    
\end{description} 

In Case 1(a) - (c), we have $r_1 < r_2$, and we let
\begin{align}\label{eq:R1}
   \hspace{-.45cm} R_1(s) = \begin{cases}
        R_{1a}(s) = \begin{cases}
            r_2 & \text{for } s > 2\sqrt{r_2}\\
            \underline r_2 & \text{for } c_1 < s \leq 2\sqrt{r_2}\\
            \underline r_1 & \text{for } s \leq c_1
        \end{cases}  &\quad \text{if } c_1 \leq 2\sqrt{r_2},\\ 
        R_{1b}(s) = \begin{cases}
            r_2 & \text{for } s > c_1,\\
            r_1 & \text{for } \lambda^* + \frac{r_1}{\lambda^*} < s \leq c_1\\
            \underline r_1 & \text{for } s \leq \lambda^* + \frac{r_1}{\lambda^*}
        \end{cases}  &\quad \text{if }  2\sqrt{r_2} <c_1< 2(\sqrt{r_1} + \sqrt{r_2-r_1}), \\
        R_{1c}(s) = \begin{cases}
            r_2 & \text{for } s > c_1,\\
            r_1 & \text{for } 2\sqrt{r_1} < s \leq c_1\\
            \underline r_1 & \text{for } s \leq 2\sqrt{r_1}
        \end{cases} &\quad \text{if } c_1 \geq 2(\sqrt{r_1} + \sqrt{r_2-r_1})
    \end{cases}
\end{align}
In Cases 2(a)-(c), we have $r_1 > r_2$, and we let
\begin{align}\label{eq:R2}
    R_2(s) = \begin{cases}
        R_{2a}(s) =  \begin{cases}
            r_2 & \text{for } s > 2\sqrt{r_2},\\
            \underline r_2 & \text{for } c_1 \leq s \leq 2\sqrt{r_2}\\
            \underline r_1 & \text{for } s < c_1.
        \end{cases}  &\quad \text{if }  c_1 <2\sqrt{r_2},\\  
      R_{2b}(s) =  \begin{cases}
            r_2 & \text{for } s > 2\sqrt{r_2},\\
            \min\{r_2,\underline r_1\} & \text{for } s = 2\sqrt{r_2}\\
            \underline r_1 & \text{for } s < 2\sqrt{r_2}.
        \end{cases}  &\quad \text{if }  c_1 = 2\sqrt{r_2},\\  
        R_{2c}(s) = \begin{cases}
            r_2 & \text{for } s \geq c_1,\\
            r_1 & \text{for } 2\sqrt{r_1} < s < c_1\\
            \underline r_1 & \text{for } s \leq 2\sqrt{r_1}
        \end{cases}  &\quad  \text{if } c_1 > 2\sqrt{r_1}.
    \end{cases}
\end{align}

\begin{lemma}
$F_*(t,x) \geq R_i(x/t)$ in cases 1(a)-(c) and 2(a)-(c).
\end{lemma}
\begin{proof}
    The lemma follows from the definition of $F_*$, in \eqref{eq:F*}, and the global bounds $\beta \leq v(t,x) \leq 1$ (Lemma \ref{lemma:gl_bds}).
\end{proof}

\section{Lower bound on the spreading speed}\label{sec:4}

We will use the Hamilton-Jacobi method to prove a lower bound for the spreading speed.
%
To this end, define the WKB-Ansatz \cite{Evans1989PDE}
\begin{align}\label{eq:w_transformation}
     w^\epsilon(t,x) = -\epsilon \log u^\epsilon(t,x) \qquad \text{ where }\quad u^\epsilon(t,x) &= u(t/\epsilon,x/\epsilon),
\end{align}
and consider the half-relaxed limits \cite{Barles1987discontinuous}
\begin{align}\label{eq:w_hr}
    w^*(t,x) = \limsup_{\substack{\epsilon \to 0 \\ (t',x') \to (t,x)}} w^\epsilon(t',x') \quad \text{and} \quad w_*(t,x) = \liminf_{\substack{\epsilon \to 0 \\ (t',x') \to (t,x)}} w^\epsilon(t',x').
\end{align}

In the following lemma, we show that  $w^*(t,x)$ and $w_*(t,x)$ can be related to one dimensional-functions $\rho^*(s)$ and $\rho_*(s)$, respectively,
\begin{lemma}\label{lemma:rho}
    Let $w^*$ and $w_*$ be defined as in \eqref{eq:w_hr}. Then $w^*(t,x) = t\rho^*(x/t)$ and $w_*(t,x) = t\rho_*(x/t)$ for some functions $\rho^*$ and $\rho_*$.
\end{lemma}
\begin{proof}
    For the existence of $\rho^*$, we may compute
\begin{align*}
    w^*(t,x) &= \limsup_{\substack{\epsilon \to 0 \\ (t',x') \to (t,x)}} -\epsilon \log u\Big(\frac{t'}{\epsilon},\frac{x'}{\epsilon}\Big) = t\limsup_{\substack{\epsilon \to 0 \\ (t'',x'') \to (1,x/t)}} -(\epsilon/t) \log u\Big(\frac{t''}{\epsilon/t},\frac{x''}{\epsilon/t}\Big).
\end{align*}
Thus $w^*(t,x)=tw^*(1,x/t)$, and the first part of the result is proved if
we take $\rho^*(s) = w^*(1,s)$. The proof of the second part is analogous.
\end{proof}

Next, we describe a bird's eye view of the Hamilton-Jacobi approach in order to achieve our final goal of bounding the spreading speed from below by the optimal constant $\sigma > 0$. For clarity, we will state the necessary lemmas and provide their proofs later on.

We start with the following lemma which is due to \cite{Evans1989PDE} for the KPP equation, the proof is presented Subsection \ref{4.1}.

\begin{lemma}\label{lemma:ulowbd}
    Suppose that there is $s_0>0$ such that $\rho^*(s) = 0$ for all $s \in [0,s_0]$. Then there exists $\delta_0 > 0$ such that
    \begin{align*}
        \liminf_{t \to \infty} \inf_{\eta t < x < (s_0-\eta) t} u(t,x) \geq \delta_0  \quad \text{ for each }\eta>0\text{ sufficiently small}.
    \end{align*}
\end{lemma}

Hence, a lower bound of the spreading speed (and hence the complete proofs of our main theorems)
can be obtained by determining the set $\{s:~\rho^*(s) = 0\}$. Precisely,  it is sufficient to show that $\rho^*(s) = 0$ for $s \in [0,\sigma]$, where $\sigma = \sigma(c_1;r_1,r_2)$ as in \eqref{eq:sigma(c_1)}.

To this end, we derive a limiting Hamilton-Jacobi equation for $w^*$ and then for $\rho^*$. Observe that $w^\epsilon$ satisfies
\begin{align}\label{eq:w_pde}
    w_t^\epsilon - \epsilon w_{xx}^\epsilon + |w_x^\epsilon|^2 + (F^\ep(t,x) - u^\epsilon ) = 0 \quad \text{ for } (t,x) \in (0,\infty) \times \mathbb R,
\end{align}
where $F^\ep$ is given in \eqref{eq:Fep}. By the fact that $F_*(t,x) \geq R(x/t)$, it is standard \cite{Evans1989PDE,lam2022asymptotic} to deduce the following.

\mycomment{
\begin{lemma}\label{lemma:4.5}
    $w^*$ is a viscosity subsolution of 
    \begin{align}\label{eq:w_hj}
   \min\{w_t + |w_x|^2 + R(x/t), w\} = 0 \quad \text{for } (t,x) \in (0,\infty) \times (0,\infty)
\end{align}
    Moreover,
    \begin{equation}
        w^*(t,0) = 0 \quad \text{ for }t \geq 0, \quad \text{ and }\quad 0\leq w^*(t,x) < +\infty \quad \text{ in }(0,\infty)\times [0,\infty).
    \end{equation}
\end{lemma}
The proof is postponed to Subsection \ref{subsection:4.2}.
}
\begin{lemma}\label{lemma:4.6}
Suppose $F_*(t,x) \geq R(x/t)$, then $\rho^*$ is a viscosity subsolution of 
    \begin{align}\label{eq:rho_hj}
   \min\{\rho - s\rho' + |\rho'|^2 + R(s), \rho\} = 0 \quad \text{for } s \in (0,\infty).
\end{align}Moreover,
\begin{equation}\label{eq:lem6b}
    \rho^*(0) = 0 \quad \text{ and }\quad \rho^*(s) <\infty \quad \text{ for all }s \in [0,\infty).
\end{equation}
\end{lemma}
\begin{proof}
We postpone the proof to Subsection \ref{subsection:4.2}.
\end{proof}

By the comparison principle, discussed in the Appendix, the Hamilton-Jacobi equation \eqref{eq:rho_hj} has a unique viscosity solution.
\begin{lemma}\label{lemma:4.7}
For any given case  $(i,j) \in \{(1,a), (1,b), (1,c), (2,a),(2,b),(2,c)\}$ as stated in Section \ref{sec:3},
let $R$ be given by $R = R_{ij}$.     
    The Hamilton-Jacobi equation \eqref{eq:rho_hj} has a unique viscosity solution, $\hat \rho$, satisfying
\begin{equation}\label{e.hatrhobc}
    \hat \rho(0) = 0 \quad \text{ and } \lim_{s \to \infty} \frac{\hat\rho(s)}{s} = \infty.
\end{equation}
Moreover, $\hat\rho$ is nondecreasing in $s$, i.e.
\begin{equation}\label{e.rmk:1}
\hat\rho(s) \equiv 0\quad \text{ for }0 \leq s \leq \sup\{s' \geq 0:~ \hat\rho(s') = 0\}.
\end{equation}
\end{lemma}

Furthermore, the following Lemma holds:
\begin{lemma}\label{lemma:4.8}
For any given case  $(i,j) \in \{(1,a), (1,b), (1,c), (2,a),(2,b),(2,c)\}$ as stated in Section \ref{sec:3},
let $R$ be given by $R = R_{ij}$ and let $\hat\rho$ be the unique solution of \eqref{eq:rho_hj} as specified in Lemma \ref{lemma:4.7}. Then 
$$
0 \leq \rho^*(s) \leq \hat\rho(s)\quad \text{ for }s \geq 0.$$ 
\end{lemma}

The statement \eqref{e.rmk:1} and Lemma \ref{lemma:4.8} imply that 
\begin{equation}
    \rho^*(s) = 0 \quad \text{ for }0 \leq s \leq \sup\{s'\geq 0:~ \hat\rho(s') = 0\}. 
\end{equation}
Together with Lemma \ref{lemma:ulowbd}, this enables us to establish
$$
\underline{c}_* \geq \sup\{s \geq 0:~ \hat\rho(s) = 0\}.
$$

Next, we will establish the explicit formula {of} the unique solution $\hat\rho$ satisfying \eqref{eq:rho_hj} in viscosity sense and \eqref{e.hatrhobc} in classical sense in Lemma \ref{lemma:explicit_soln}.

For example, in Case 1(a) ($r_1< r_2$ and $c_1 \leq 2\sqrt{r_2}$), 
we find that
\begin{align}
        \hat \rho(s) = \begin{cases} 
            \frac{s^2}{4}-r_2 \quad &\text{for } s > 2\sqrt{r_2}\\
            0 \quad &\text{for } 0 \leq s \leq 2\sqrt{r_2}.
        \end{cases}
        \end{align}
Hence, we deduce from Lemma \ref{lemma:4.8} that 
$$
0 \leq \rho^*(s) \leq \hat\rho(s) = 0 \quad \text{ for } 0\leq s \leq 2\sqrt{r_2}.
$$
By Lemma \ref{lemma:ulowbd}, we conclude that $\underline{c}_* \geq 2\sqrt{r_2}$. This establishes the lower bound of the spreading speed in Case 1(a). The spreading speed in Case 1(a) is thus determined, since $2\sqrt{r_2}$ is also the upper bound of spreading speed (thanks to Proposition \ref{lemma:upper_speed}).

In the next couple subsections, we present the proofs of the above lemmas.

\subsection{Proof of Lemma \ref{lemma:ulowbd}}\label{4.1}

\begin{proof}[Proof of Lemma \ref{lemma:ulowbd}]
Our proof is adapted from Theorem 1.1 of \cite{Evans1989PDE}.  Fix a small $0<\eta\ll 1$. It is sufficient to show that there exists $\delta_0=\delta_0(\eta) > 0$ such that
\begin{align}\label{eq:inf_K}
\liminf_{\epsilon \to 0}\inf_{K} u^\epsilon(t,x) \geq \delta_0
\end{align}
for any compact set given by
$K = \{(1,x) : \eta \leq x \leq s_0-\eta\} \subset \subset \{(t,x) : 0 < x/t < s_0\}$.

Indeed, 
\begin{align*}
    \liminf_{t \to \infty}\inf_{\eta t < x < (s_0-\eta)t} u(t,x) =  \liminf_{\epsilon \to 0}\inf_{K} u^\epsilon(t,x) \geq \delta_0. 
\end{align*}

To show \eqref{eq:inf_K}, we first observe that $w^*(t,x) = t\rho^*(x/t) = 0$ in some compact subset $\tilde{K}$ such that 
$$
K \subset {\rm Int}\,\tilde K \subset \{(t,x):~0<x/t<s_0\},
$$
which implies $w^\epsilon(t,x) \to 0$ uniformly in a neighborhood of $K$. Now for $(t_0,x_0) \in K$, let $\psi(t,x) = (t-t_0)^2 + (x-x_0)^2$. Then $w^* - \psi$ has a strict local maximum at $(t_0,x_0)$. Since $w^\epsilon \to 0$ uniformly in a neighborhood of $K$, for each $\epsilon> 0$ sufficiently small, the function $w^\epsilon - \psi$ has a local maximum at $(t_\epsilon,x_\epsilon) \in K$, where $(t_\epsilon,x_\epsilon) \to (t_0,x_0)$ as $\epsilon \to 0$. Thus,  
\begin{align}
    o(1) = \partial_t\psi - \epsilon\partial_{xx}\psi + |\partial_x\psi|^2 \leq \partial_t w^\epsilon - \epsilon\partial_{xx}w^\epsilon + |\partial_xw^\epsilon|^2 = u^\epsilon - F^\epsilon \leq u^\epsilon - \delta_0 \label{eq:ssss1}
\end{align} 
at $(t,x)=(t_\epsilon,x_\epsilon)$, where $\delta_0 = \beta\inf_{s \in \mathbb R} a(s)-1 > 0$.

Using the fact that $w^\epsilon - \psi$ has a local maximum at $(t_\epsilon,x_\epsilon)$, we deduce that
\begin{align*}    w^\epsilon(t_\epsilon,x_\epsilon)\geq (w^\epsilon-\psi)(t_\epsilon,x_\epsilon)  \geq  (w^\epsilon-\psi)(t_0,x_0)
=w^\epsilon(t_0,x_0)
\end{align*}
which implies that 
$u^\epsilon(t_0, x_0) \geq u^\epsilon(t_\epsilon,x_\epsilon)$. Combining with \eqref{eq:ssss1}, we have
\begin{align*}
    u^\epsilon(t_0,x_0) \geq u^\epsilon(t_\epsilon,x_\epsilon) \geq \delta_0 + o(1).
\end{align*}
Since the above argument is uniform for arbitrary $(t_0,x_0) \in K$, this implies \eqref{eq:inf_K}.
\end{proof} 

\subsection{Proof of Lemma \ref{lemma:4.6}}\label{subsection:4.2}
We recall the definition of viscosity solutions of \eqref{eq:rho_hj}. We recall the definition of viscosity solution of Hamilton-Jacobi equations with discountinuous Hamiltonian, following \cite{Barles2013introduction} and originating from \cite{Ishii1985hamilton} (see also \cite{Evans1989PDE,lam2022asymptotic} for the definition involving variational inequalities).
\begin{definition}
 In the following let $R^*$ and $R_*$ be the upper and lower envelope of $R$, which is given by
$$
R^*(s) = \limsup_{s'\to s} R(s') \quad \text{ and }\quad R_*(s) = \liminf_{s'\to s} R(s') .
$$
    
   \begin{itemize} 
   \item A lower semicontinuous function $\hat \rho$ is called a viscosity super-solution of \eqref{eq:rho_hj} if $\hat \rho \geq 0$, and for any test function $\phi \in C^1$, if $s_0$ is a strict local minimum of $\hat \rho - \phi$, then 
    \begin{align*}
        \hat \rho(s_0) - s_0\phi'(s_0) + |\phi'(s_0)|^2 + R^*(s_0) \geq 0.
    \end{align*}
    \item An upper semicontinuous function $\hat \rho$ is called a viscosity sub-solution of \eqref{eq:rho_hj} if for any test function $\phi \in C^1$, if $s_0$ is a strict local maximum of $\hat \rho - \phi$ and $\hat \rho(s_0) > 0$, then 
    \begin{align*}
        \hat \rho(s_0) - s_0\phi'(s_0) + |\phi'(s_0)|^2 + R_*(s_0) \leq 0.
    \end{align*}
   \item  We say $\hat \rho$ is a viscosity solution of \eqref{eq:rho_hj} if $\hat \rho$ is a viscosity super- and sub-solution.
   \end{itemize}
\end{definition}

We will first show that $\rho^*$ is nonnegative and that $\rho^*(0) = 0$.

\begin{lemma}\label{lemma:rho0}
    Let $\rho^*$ be defined as in Lemma \ref{lemma:rho}. Then
    \begin{align}\label{eq:rho_properties}
        \rho^*(0) = 0 \quad \text{and} \quad \rho^*(s) \geq 0 \text{ for } s \geq 0.
    \end{align}
\end{lemma}
\begin{proof}
    We first show $w^*(t,x) \geq 0$ for $(t,x) \in (0,\infty) \times \mathbb R$. Indeed, since $u(t,x) \leq \max\{r_2,r_1\}$ for all $(t,x) \in [0,\infty) \times \mathbb R$, by the definition of $w^\epsilon$ we have $w^\epsilon \geq -\epsilon \log(\max\{r_2,r_1\})$ for each $\epsilon > 0$ and $(t,x) \in (0,\infty) \times \mathbb R$, and we may compute 
    \begin{align}
    w^*(t,x) = \limsup_{\substack{\epsilon \to 0 \\ (t',x') \to (t,x)}} w^\epsilon(t',x') \geq 0 \quad \text{ for } (t,x) \in (0,\infty) \times \mathbb{R}. \label{eq:ss1s}
    \end{align}    
    In particular, $w^*(t,0) \geq 0$ for $t > 0$.
    
    The proof will be complete once we show $w^*(t,0) \leq 0$ for $t >0$. Denote $\underline r = \underline r_1 \wedge \underline r_2$. Then using the lower bound $v \geq \beta$, we see that $u$ is a super-solution of
\begin{align}
    \underline{u}_t - \underline{u}_{xx} &= \underline u(\underline r - \underline u) \quad \text{ in } (0,\infty) \times \mathbb{R}. \label{eq:diffu_log}
\end{align}
Let $\underline u(t,x)$ be the solution of \eqref{eq:diffu_log} with identical (compactly supported) initial {condition} as $u(t,x)$, then 
 the classical spreading result for the diffusive logistic equation \cite{aronson1978multidimensional} says that $\underline u$ has spreading speed $2\sqrt{\underline{r}}$. In particular,
 $$
 \liminf_{t\to\infty} \inf_{|x|<\sqrt{\underline r}t} \underline{u}(t,x)\geq 2\delta_1 \quad \text{ for some }\delta_1>0.
 $$
By the comparison principle, $u \geq \underline u$, i.e. there exists $t_1 > 0$ such that
\begin{align*}
    \inf_{|x| < \sqrt{\underline r}t} u(t,x) \geq \delta_1 \quad \text{ for }t \geq t_1,
\end{align*}
which implies
\begin{align}\label{eq:wupperbd}
    \sup_{\substack{|x| < \sqrt{\underline r}t\\t\geq \epsilon t_1}} w^\epsilon(t,x) \leq  -\epsilon \log \delta_1.
\end{align}
Now, fix an arbitrary $t_0>0$. Let $(t,x) \to (t_0,0)$ and $\ep \to 0$, we deduce
\begin{align}
    w^*(t_0,0) = \limsup_{\substack{\epsilon \to 0 \\ (t,x) \to (t_0,0)}} w^\epsilon(t,x) \leq 0. \label{eq:ss2s}
\end{align}
Combining \eqref{eq:ss1s} and \eqref{eq:ss2s}, we have
\begin{equation}\label{eq:ss3s}
    w^*(t,0) = 0 \quad \text{ for all }t>0.
\end{equation}
We recall $w^*(t,x) = t\rho^*(x/t)$ (thanks to Lemma \ref{lemma:rho}), so that \eqref{eq:rho_properties} directly follows from \eqref{eq:ss1s} and \eqref{eq:ss3s}.
This completes the proof.
\end{proof}
The following lemma implies that $\rho^*(s) < \infty$ for $s \in [0,\infty)$.
\begin{lemma}\label{lemma:wlocalbounds}
    Let $w^\epsilon$ be a solution of \eqref{eq:w_pde}. Then for each compact subset $Q$ of $(0,\infty) \times \mathbb{R}$, there is a constant $C(Q)$ independent of $\epsilon$ such that
    \begin{align*}
        w^\epsilon(t,x) \leq C(Q) \quad \text{ for } (t,x) \in Q \text{ and } \epsilon \in (0,1/C(Q)].
    \end{align*}
In particular,
\begin{equation}\label{eq:finite1}
    w^*(t,x) < +\infty \quad \text{ for each }(t,x) \in (0,\infty)\times \mathbb{R} \quad \text{ and }\quad \rho^*(s) <+\infty \quad \text{ for each }s\in [0,\infty).
\end{equation}
\end{lemma}
\begin{proof}
We only prove the bound for $Q \subset (0,\infty)\times [0,\infty)$. The case for $Q \subset (0,\infty)\times (-\infty,0]$ is similar and is omitted. 
    Our proof follows the ideas in \cite{Evans1989PDE}. Fix $\delta \in (0,1)$ such that $Q \subset [\delta, 1/\delta] \times [0, 1/\delta]$. Define 
    \begin{align*}
        z^\epsilon(t,x) = \frac{|x+2\delta|^2}{4t} + \frac{\epsilon}{2}\log t + C_\delta(1+t).
    \end{align*}
    By taking $C_\delta > 0$ to be a large constant depending on $\delta$,  $z^\epsilon$ is a (classical) super-solution of \eqref{eq:w_pde} in $(0,\infty) \times (0,\infty)$. 

    By \eqref{eq:ss3s} in the proof of Lemma \ref{lemma:rho0} and the definition of $w^*$, there is a constant $C_\delta > 0$ such that, 
    \begin{align*}
        \sup_{0 < \epsilon \leq 1/2} w^\epsilon(t+\delta/2,0) \leq {C_\delta} \quad \text{ for } t \in [0, 1/\delta].
    \end{align*}
Observe that for $\epsilon$ sufficiently small, we have
    \begin{align*}
        \begin{cases}
            w^\epsilon(\delta/2,x)< \infty = z^\epsilon(0,x) \quad \text{ for } x \geq 0\\
            w^\epsilon(t+\delta/2,0) \leq C_\delta \leq z^\epsilon(t,0) \quad \text{ for } t \in [0,1/\delta].
        \end{cases}
    \end{align*}
    It follows from the maximum principle that 
    \begin{align*}
        w^\epsilon(t+\delta/2,x) \leq z^\epsilon(t,x) \quad \text{ for } (t,x) \in [0,1/\delta] \times [0, \infty).
    \end{align*}
Taking supremum over $[\delta/2,1/\delta]\times [0,1/\delta]$, we have
\begin{equation}
\sup_{[\delta/2,1/\delta]\times [0,1/\delta]}    w^\epsilon(t+\delta/2 ,x) \leq C'_\delta:=\sup_{[\delta/2,1/\delta]\times [0,1/\delta]} \left[\frac{|x+2\delta|^2}{4t} + \log t  + C_\delta(1+t) \right].
\end{equation}
This completes the proof.
\end{proof}


Next, we prove Lemma \ref{lemma:4.6}.

\begin{proof}[Proof of Lemma \ref{lemma:4.6}]
Since \eqref{eq:lem6b} is a consequence of Lemma \ref{lemma:rho0} and \eqref{eq:finite1}, it remains to show that $\rho^*$ is a viscosity subsolution of \eqref{eq:rho_hj}.

Let $\phi \in C^1$ be a test function and suppose that $\rho^* - \phi$ has a strict local maximum at $s = s_0$, and that $\rho^*(s_0) > 0$. Without loss of generality, we may assume that $\rho^* - \phi \leq 0$ for all $s$ near $s_0$, with equality holding only at $s = s_0$. We will show that
\begin{align*}
        \rho^*(s_0) - s_0\phi'(s_0) + |\phi'(s_0)|^2 + R_*(s_0) \leq 0.
\end{align*}
(Note that $R_*(s_0)= R(s_0)$ by our definition of $R$ in \eqref{eq:R1} and \eqref{eq:R2}.) 
First, we note that $w^*(t,x) = \rho^*(x/t)$ and that $w^*(t,x) - t\phi(x/t) - (t-1)^2 \leq 0$ for all $(t,x)$ near $(1,s_0)$, with equality holding only at $(t,x) = (1,s_0)$. Define, in terms of $\phi$, a two variable test function $$\varphi(t,x) = t\phi(x/t) - (t-1)^2.$$ Then by the definition of $w^*$, there exists a sequence $\epsilon_n \to 0$ and sequence of points $(t_n,x_n) \to (1,s_0)$ as $n \to \infty$ such that: $w^{\epsilon_n} - \varphi$ has a local maximum at $(t_n,x_n)$, and $w^{\epsilon_n}(t_n,x_n) \to w^*(t_0,x_0) > 0$. Thus, for $(t,x) = (t_n,x_n)$, we have
\begin{align*}
    \partial_t\varphi = \partial_t w^{\epsilon_n} &= \epsilon_n \partial_{xx}w^{\epsilon_n} - |\partial_x w^{\epsilon_n}|^2 - (F^{\epsilon_n} - u^{\epsilon_n})\\
    &\leq \epsilon_n\partial_{xx}\varphi - |\partial_x \varphi|^2 - (F^{\epsilon_n} - u^{\epsilon_n}).
\end{align*}
Thus, 
\begin{align}\label{eq:phi_sub}
   \partial_t\varphi - \epsilon_n\partial_{xx}\varphi + |\partial_x \varphi|^2 + (F^{\epsilon_n} - u^{\epsilon_n}) \leq 0
\end{align}
for $(t,x) = (t_n,x_n)$.
Letting $n \to \infty$, we obtain
\begin{align}\label{eq:sub_sol_step}
   \partial_t\varphi(t_0,x_0) + |\partial_x \varphi(t_0,x_0)|^2 + F_*(t_0,x_0) \leq 0,
\end{align}
where we have used the fact that $w^{\epsilon_n}(t_n,x_n) \to w^*(t_0,x_0) > 0$ implies $u^{\epsilon_n}(t_n,x_n) \to 0$. Since $F_*(t_0,x_0) \geq R(x_0/t_0)$, it follows from \eqref{eq:sub_sol_step} that
\begin{align*}
    \phi(s_0) - s_0\phi'(s_0)  + |\phi'(s_0)|^2 + R_*(s_0) \leq 0.
\end{align*} 
Thus, $\rho^*$ is a viscosity sub-solution of \eqref{eq:rho_hj}.
%
\end{proof}

\subsection{Proof of Lemmas \ref{lemma:4.7} and \ref{lemma:4.8}}\label{section:rho}
\begin{proof}[Proof of Lemma \ref{lemma:4.7}]
Observe, in each case $i = 1,2$, that our choice of $R$ satisfies (B1)-(B2) of the Appendix. By Corollary \ref{corollary:uniqueness} of the Appendix, there exists a unique $\hat\rho$ satisfying \eqref{eq:rho_hj} in the viscosity sense, and the boundary condition \eqref{e.hatrhobc} in the classical sense. Moreover, $s\mapsto \hat\rho$ is nondecreasing. This proves Lemma \ref{lemma:4.7}.
\end{proof}

\begin{proof}[Proof of Lemma \ref{lemma:4.8}]
We observe that $\rho^*$ is a viscosity subsolution (by Lemma \ref{lemma:4.6}) and that $\hat\rho$ is a viscosity super-solution (by Lemma \ref{lemma:4.7}). Moreover, 
by \eqref{eq:lem6b} and \eqref{e.hatrhobc}, we have
$$
\rho^*(0) = \hat\rho(0)=0,\quad \text{ and }\quad \lim_{s \to \infty} \frac{\rho^*(s)}{s} \leq +\infty =\lim_{s \to \infty} \frac{\hat \rho(s)}{s}.
$$
We can therefore apply the comparison principle (see Lemma \ref{lemma:comp} of the Appendix) to derive 
$$
\rho^*(s) \leq \hat\rho(s) \quad \text{ for }s \geq 0.
$$
Finally, $\rho^*(s) \geq 0$ is proved in Lemma \ref{lemma:rho0}.
\end{proof}


\section{Solving for the spreading speed via explicit formulas for $\hat\rho$} \label{sec:explicitform}
For each of the 
cases 1(a)-(c), 2(a)-2(c), we will propose an explicit formula for $\hat\rho$ in Subsection \ref{subsec:5.1}. Thanks to the uniqueness result in Lemma \ref{lemma:4.7}, it is enough to verify  (separately for each of the cases) that the given expression defines a viscosity solution of \eqref{eq:rho_hj}. This will be done in Subsection \ref{subsec:5.2}. 

\subsection{Explicit formulas for $\hat\rho$}\label{subsec:5.1}

Below, we state the explicit formula for $\hat \rho$ in each case. Subsequently, we will verify in Lemma \ref{lemma:explicit_soln} that $\hat \rho$ solves \eqref{eq:rho_hj} by invoking the definition of the viscosity solution \cite{Barles2013introduction}. 

\begin{itemize}
    \item{\bf Case 1(a)}: $r_1< r_2$ and $c_1 \leq 2\sqrt{r_2}$. 
        \begin{align}\label{eq:rho1a}
        \hat \rho(s) = \begin{cases} 
            \frac{s^2}{4}-r_2 \quad &\text{for } s > 2\sqrt{r_2}\\
            0 \quad &\text{for } 0 \leq s \leq 2\sqrt{r_2},
        \end{cases}
        \end{align}
    \item{\bf Case 1(b)}: $r_1 < r_2$ and $2\sqrt{r_2} < c_1 < 2(\sqrt{r_1} + \sqrt{r_2-r_1})$.
        \begin{align}\label{eq:rho1b}
        \hat \rho(s) = \begin{cases}
            \frac{s^2}{4}-r_2 \quad &\text{for } s > c_1\\
            \lambda^* s - \big((\lambda^*)^2 + r_1\big) \quad &\text{for } (\lambda^* + \frac{r_1}{\lambda^*}) < s \leq c_1\\
            0 \quad &\text{for } 0 \leq s \leq \lambda^* + \frac{r_1}{\lambda^*},
        \end{cases}
        \end{align}
    where $\lambda^* = \frac{c_1}{2} - \sqrt{r_2-r_1}$.
    \item {\bf Case 1(c)}: $r_1< r_2$ and $c_1 \geq 2(\sqrt{r_1} + \sqrt{r_2-r_1})$
    \begin{align}\label{eq:rho1c}
        \hat \rho(s) = \begin{cases}
            \frac{s^2}{4}-r_2 \quad &\text{for } s > c_1\\
            \lambda^* s - \big((\lambda^*)^2 + r_1\big) \quad &\text{for } 2\lambda^* < s \leq c_1\\
            \frac{s^2}{4}-r_1 \quad &\text{for } 2\sqrt{r_1} < s \leq 2\lambda^*  \\
            0 \quad &\text{for } 0 \leq s \leq 2\sqrt{r_1}.
        \end{cases}
    \end{align}

    \item {\bf Case 2(a)}: $r_1 > r_2$ and $c_1 < 2\sqrt{r_2}$
        \begin{align}\label{eq:rho2a}
        \hat \rho(s) = \begin{cases} 
            \frac{s^2}{4}-r_2 \quad &\text{for } s > 2\sqrt{r_2}\\
            0 \quad &\text{for } 0 \leq s \leq 2\sqrt{r_2},
        \end{cases}
        \end{align}
    \item {\bf Case 2(b)}: $r_1 > r_2$ and $c_1 = 2\sqrt{r_2}$ 
        \begin{align}\label{eq:rho2b}
        \hat \rho(s) = \begin{cases} 
            \frac{s^2}{4}-r_2 \quad &\text{for } s > 2\sqrt{r_2}\\
            0 \quad &\text{for } 0 \leq s \leq 2\sqrt{r_2},
        \end{cases}
        \end{align}
    \item {\bf Case 2(c)}: $r_1 > r_2$ and $c_1 > 2\sqrt{r_1}$.
        \begin{align}\label{eq:rho2c}
        \hat \rho(s) = \begin{cases} 
            \frac{s^2}{4}-r_2 \quad &\text{for } s > 2\tilde \lambda \\
            \tilde \lambda s - (\tilde \lambda^2 + r_2) &\text{for } c_1 < s \leq 2\tilde \lambda \\
            \frac{s^2}{4}-r_1 \quad &\text{for } 2\sqrt{r_1} < s < c_1\\
            0 \quad &\text{for } 0 \leq s \leq 2\sqrt{r_1},
        \end{cases}
    \end{align}
        where $\tilde \lambda = \frac{c_1}{2} + \sqrt{r_1-r_2}$. 
 \end{itemize}

 \subsection{$\hat\rho$ solves the HJE \eqref{eq:rho_hj} in the viscosity sense}\label{subsec:5.2}
 
\begin{lemma}\label{lemma:explicit_soln}
The unique viscosity solution $\hat\rho$ of \eqref{eq:rho_hj} satisfying (in the classical sense)
$$
\hat\rho(0) =0 \quad \text{ and }\quad \lim_{s \to \infty}\frac{\hat\rho(s)}{s} = \infty,
$$
is given by the formulas of Subsection \ref{subsec:5.1}.
\end{lemma}
\begin{proof}
We prove that the above formulas determine the the unique viscosity solution guaranteed by Lemma \ref{lemma:4.7}. To do that, it is enough to show that, in each case, $\hat\rho$ as given above, satisfies (i) the Hamilton-Jacobi equation \eqref{eq:rho_hj} in the viscosity sense, as well as (ii) the boundary condition \eqref{e.hatrhobc} in the classical sense. In view of the explicit formulas, (ii) is obvious. Thus, it remains to verify that $\hat\rho$ is a viscosity solution of \eqref{eq:rho_hj} in each case. Since the proof only differs slightly in each case, we consider only two representative cases 1(b) and 2(c) here, and omit the verification of the rest. 

Let us proceed with Case 1(b), where we fix $r_1, r_2, c_1$ satisfying
\begin{equation}\label{cond:1b}
r_1 < r_2,\quad \text{ and }\quad  2\sqrt{r_2} < c_1 < 2(\sqrt{r_1} + \sqrt{r_2-r_1}).
\end{equation}
Next, we  set
$$
R(s) = R_{1b}(s) = \begin{cases}
            r_2 & \text{for } s > c_1,\\
            r_1 & \text{for } \lambda^* + \frac{r_1}{\lambda^*} < s \leq c_1\\
            \underline r_1 & \text{for } s \leq \lambda^* + \frac{r_1}{\lambda^*}.
        \end{cases}
$$
and 
    \begin{equation}\label{eq:hatrho1b}
        \hat \rho(s) = \begin{cases}
            \frac{s^2}{4}-r_2 \quad &\text{for } s > c_1\\
            \lambda^* s - \big((\lambda^*)^2 + r_1\big) \quad &\text{for } (\lambda^* + \frac{r_1}{\lambda^*}) < s \leq c_1\\
            0 \quad &\text{for } 0 \leq s \leq \lambda^* + \frac{r_1}{\lambda^*},
        \end{cases}
    \end{equation}
    where $\lambda^* = \frac{c_1}{2} - \sqrt{r_2-r_1}$. 

First, observe that $\hat\rho$ is continuous, thanks to our choice of $\lambda^*$.

Next, we show that $\hat\rho$ is a viscosity subsolution of \eqref{eq:rho_hj}. 
To this end, observe that $\hat\rho$ satisfies the equation \eqref{eq:rho_hj} in the classical sense almost everywhere in $[0,\infty)$. (In fact, it satisfies the equation classically 
for $s \in \mathbb{R}\setminus\{c_1, \lambda^* + r_1/\lambda^*\}$.)  
By the convexity of the Hamiltonian, we can apply \cite[Proposition
5.1]{Bardi1997optimal} to conclude that it is in fact a viscosity sub-solution of \eqref{eq:rho_hj}. 
    
Next, we show that $\hat \rho$ is a viscosity super-solution of \eqref{eq:rho_hj}. Suppose $\hat \rho - \phi$ obtains a strict local minimum at $s_0 \in [0,\infty)$ for some $\phi \in C^1$. Now, $\hat \rho$ is a classical solution of \eqref{eq:rho_hj} for all $s \not \in \{\lambda^* + \frac{r_1}{\lambda^*},c_1\}$, so it automatically satisfies  \eqref{eq:rho_hj} in the viscosity sense. 
We need only consider $s_0 \in \{\lambda^* + \frac{r_1}{\lambda^*},c_1\}$. Suppose $s_0 = \lambda^* + \frac{r_1}{\lambda^*}$. Then $0\leq \phi'(s_0) \leq \lambda^*$, and $R^*(s_0) = \max\{\underline r_1,r_1\} = r_1$. Therefore, at the point $s=s_0$, it holds that
    \begin{align*}
        \rho - s_0\phi' + |\phi'|^2 + R^* &=   - \big(\lambda^* + \frac{r_1}{\lambda^*}\big)\phi'+ |\phi'|^2 + r_1= (\phi' - \lambda^*)(\phi' - \frac{r_1}{\lambda^*}) \geq 0, 
    \end{align*}
    where the last inequality is a consequence of $\phi'(s_0) \leq \lambda^* < \frac{r_1}{\lambda^*}$ (which in turn follows from the choice of $\lambda^*$ and the condition {\eqref{cond:1b}}). 
    
    If $s_0 = c_1$, then $R^*(s_0) = \max\{r_1,r_2\} = r_2$, and we have
    \begin{align*}
        \rho - s_0\phi' + |\phi'|^2 + R^* &= \left(\tfrac{c_1^2}{4} - r_2\right) - c_1\phi' + |\phi'|^2 + r_2 = (\phi' - \frac{c_1}{2})^2 \geq 0 \quad \text{ at }s = s_0.
    \end{align*} 
    This proves that $\hat \rho$ is a viscosity super-solution.

This completes the proof that the unique viscosity solution $\hat\rho$ as guaranteed by Lemma \ref{lemma:4.7} is given by the explicit formula \eqref{eq:hatrho1b} for the first representative Case 1(b).
    

Let us proceed with Case 2(c), where we fix $r_1, r_2, c_1$ satisfying
\begin{equation}\label{cond:2b}
r_1 > r_2,\quad \text{ and }\quad  c_1 > 2\sqrt{r_1}.
\end{equation}
Next, we  set
$$
R(s) = R_{2c}(s) = \begin{cases}
            r_2 & \text{for } s \geq c_1,\\
            r_1 & \text{for } 2\sqrt{r_1} < s < c_1\\
            \underline r_1 & \text{for } s \leq 2\sqrt{r_1}
        \end{cases}
$$
and 
    \begin{equation}\label{eq:hatrho2c}
        \hat \rho(s) = \begin{cases} 
            \frac{s^2}{4}-r_2 \quad &\text{for } s > 2\tilde \lambda \\
            \tilde \lambda s - (\tilde \lambda^2 + r_2) &\text{for } c_1 < s \leq 2\tilde \lambda \\
            \frac{s^2}{4}-r_1 \quad &\text{for } 2\sqrt{r_1} < s \leq c_1\\
            0 \quad &\text{for } 0 \leq s \leq 2\sqrt{r_1},
        \end{cases}
    \end{equation}
    where $\tilde\lambda = \frac{c_1}{2} + \sqrt{r_1-r_2}$. 

Again, we first observe that $\hat\rho$, as given in \eqref{eq:hatrho2c}, is continuous, and satisfies the equation \eqref{eq:rho_hj} in the classical sense for $s \in [0,\infty)\setminus\{c_1,2\sqrt{r_1}\}$ (note that it is in fact continuously differentiable in a neighborhood of $s=2\tilde\lambda$).  It then follows again from \cite[Proposition
5.1]{Bardi1997optimal}  that $\hat\rho$ is a viscosity sub-solution of \eqref{eq:rho_hj}. 

Next, we verify that it is also a viscosity super-solution. Suppose $\hat \rho - \phi$ obtains a strict local minimum at $s_0 \in [0,\infty)$ for some $\phi \in C^1$. Since $\hat \rho$ is a classical solution of \eqref{eq:rho_hj} for all $s \not \in \{c_1,2\sqrt{r_1}\}$, we need only consider $s_0 = c_1$ or $s_0=2\sqrt{r_1}$.  
Suppose $s_0 = c_1$, then 
$$
\rho - s_0 \phi' + |\phi'|^2 + R^* = \left(\tfrac{c_1^2}{4} - r_1\right) - c_1 \phi' + |\phi'|^2 + r_1 = \left( \phi' - \tfrac{c_1}{2}\right)^2 \geq 0 \quad \text{ at }s=s_0.
$$
Suppose $s_0 = 2\sqrt{r_1}$, then 
$$
\rho - s_0 \phi' + |\phi'|^2 + R^* = 0 - (2\sqrt{r_1}) \phi' + |\phi'|^2 + r_1 = \left( \phi' - \sqrt{r_1}\right)^2 \geq 0 \quad \text{ at }s = s_0.
$$
This verifies that $\hat\rho$ is a viscosity super-solution of \eqref{eq:rho_hj}. This completes the proof that the unique viscosity solution $\hat\rho$ as guaranteed by Lemma \ref{lemma:4.7} is given by the explicit formula \eqref{eq:hatrho2c} for the first representative Case 2(c). 

We omit the verification of the other cases since they are analogous.
\end{proof}

\begin{corollary}\label{corollary:c_low_bd}
Suppose either that (1) $r_1 < r_2$, or that (2) $r_1 > r_2$ and $c_1 \not\in (2\sqrt{r_1},2\sqrt{r_2}]$. Then there exists $\delta_0>0$ such that for each $\eta>0$, 
    \begin{align*}
        \liminf_{t \to \infty} \inf_{{|x|} < (\sigma-\eta) t} u(t,x) \geq \delta_0  \quad \text{ for each }\eta>0\text{ small enough},
    \end{align*}
    where $\sigma=\sigma(c_1;r_1,r_2)$ is given in \eqref{eq:sigma(c_1)}. 
    In particular, we have
    \begin{equation}
        \underline{c}_* \geq \sigma(c_1;r_1,r_2).
    \end{equation}
\end{corollary}
\begin{proof}
Observe that
$$
u_t \geq  u_{xx} + u(\delta' - u) \quad \text{ for }(t,x) \in (0,\infty)\times \mathbb{R} 
$$
where $\delta' = \beta \inf_{s \in \mathbb R} a(s) - 1>0$, thanks to (H1). Since $u$ has compactly supported initial data, it follows from standard theory that the spreading speed of $u$ is bounded from below by $2\sqrt{\delta'}$, i.e.
 \begin{equation}\label{eq:s.1}
        \liminf_{t \to \infty} \inf_{|x| < 3\sqrt{\delta'} t/2} u(t,x) \geq \delta_0.
\end{equation}

For given $c_1,r_1,r_2>0$ satisfying any of the cases 1(a)-(c) and 2(a)-(c), Lemma \ref{lemma:explicit_soln} says that the unique solution $\hat\rho$ guaranteed in Lemma \ref{lemma:4.7} is given as in Subsection \ref{subsec:5.1}. If we  define
$$
\hat{s}:= \sup\{ s\geq 0:~ \hat\rho(s)=0\},
$$
then it is easy to see that $\hat{s} = {\sigma}(c_1;r_1,r_2)$, where the latter is given in \eqref{eq:sigma(c_1)}.
 By Lemmas \ref{lemma:4.8} and \ref{lemma:rho0}, $0 \leq \rho^*(s) \leq \hat\rho(s)$ for $s \geq 0$. Thus, $\rho^*(s) = 0$ for all $s \in [0,\sigma]$, where $\sigma = \sigma(c_1;r_1,r_2)$.

    It follows from Lemma \ref{lemma:ulowbd} that 
     \begin{equation}\label{eq:s.2}
        \liminf_{t \to \infty} \inf_{ {\eta t <} x < (\sigma-\eta) t} u(t,x) \geq \delta_0  \quad \text{ for each }\eta>0\text{ small enough},
\end{equation}
and a similar statement holds for $x <0$. The desired result follows by combining with \eqref{eq:s.1}.
\end{proof}

\begin{proof}[Proof of Theorems \ref{thm:increasing} and Theorem \ref{thm:decreasing}]
We recall that $\underline c_* \leq \overline c_*$, by construction. On the other hand, by Lemma \ref{lemma:upper_speed}, $\overline c_* \leq \sigma_1(c_1;r_1,r_2)$, and by Corollary \ref{corollary:c_low_bd}, $\underline c_* \geq \sigma_1(c_1;r_1,r_2)$. It follows that $\underline c_* = \overline c_* = \sigma_1(c_1;r_1,r_2)$, so that the spreading speed of $u$ is given by $c^* = \sigma_1(c_1;r_1,r_2)$.
\end{proof}

\section{Convergence to homogeneous state}\label{sec:homo}
In this section, we apply the previous spreading result to characterize the long-time behavior of solutions of \eqref{eq:system} in the moving frame where the predator persists.

Having established the existence of the spreading speed $c^*$, and recalling Lemma \ref{lemma:prey_spd}, it follows that $(u,v) \to (0,1)$ locally uniformly in any moving frames with speed above $c^*$, and that $u$ persists in any moving frames with speed below $c^*$. In the following we discuss the asymptotic behavior of the solutions in the latter case.  Define, for $i=1,2$,
\begin{align}\label{eq:steadystate1}
    u_i = \frac{a_i - 1}{1+ a_ib}, \quad v_i = \frac{1+b}{1+a_ib}
\end{align}
where $a_1 \coloneqq a(-\infty)$ and $a_2  \coloneqq a(+\infty)$. Then $(u_i,v_i)$ 
is the unique positive root of the algebraic system
$$
u(-1 -u+a_iv) = 0 = rv(1-v-bu).
$$
Then Theorem \ref{thm:main3} can be restated as follows. 
\begin{theorem}\label{theorem:steadystate}
    Let $(u,v)$ be the solution of \eqref{eq:system}, where $(u_0,v_0) \in C^2(\mathbb R)$ satisfies (IC). 
    \begin{itemize}
        \item[{\rm(a)}] Suppose $c_1 \geq c^*$. For any $\eta \in (0, c^*/2)$,
        \begin{align}\label{eq:ss1}
            \lim_{t \to \infty}\sup_{|x| < (c^* - \eta)t} \norm{(u,v) - (u_1,v_1)} = 0.
        \end{align}
        where  $\norm{\cdot}$ denotes the Euclidean norm in $\mathbb{R}^2$.
        \item[{\rm(b)}] Suppose $c_1 < c^*$. For any $\eta \in (0, c_1/2)$,
        \begin{align}\label{eq:ss2}
            \lim_{t \to \infty}\sup_{|x| < (c_1 - \eta)t} \norm{(u,v) - (u_1,v_1)} = 0,
        \end{align}
        and for any $\eta \in \big(0, (c^*-c_1)/2\big)$,
        \begin{align}\label{eq:ss3}
            \lim_{t \to \infty}\sup_{(c_1 + \eta)t < x < (c^* - \eta)t} \norm{(u,v) - (u_2,v_2)} = 0,
        \end{align}
        where $(u_i,v_i)$is defined in \eqref{eq:steadystate1}.
    \end{itemize}
\end{theorem}
\begin{proof}[Proof of Theorem \ref{theorem:steadystate}]
Denote $c^* = \sigma(c_1;r_1,r_2)$. 
Thanks to Corollary \ref{corollary:c_low_bd}, for each $\eta>0$, there exists $T=T(\eta)>0$ and $\delta_0>0$ such that
\begin{equation}\label{e.ss3a}
    u(t,x) \geq \delta_0 \quad \text{ for }t \geq T,~-\eta t\leq x\leq (c^* - \eta/2)t.
\end{equation}

Suppose for a contradiction that there exists $\epsilon_0,\eta > 0$ and a sequence $(t_k,x_k)$ with $t_k \to \infty$ and $0\leq x_k < (c^* -\eta)t_k$ (the case for $x_k \leq 0$ is similar), such that $\norm{(u,v)(t_k,x_k) - (\tilde u,\tilde v)} > \epsilon_0$ for all $k \geq 1$, where 
\begin{equation}\label{e.ss.35}
    (\tilde u,\tilde v) = \begin{cases}
        (u_2,v_2) &\text{ if }c_1<c^*\quad \text{ and }x_k \in \big((c_1+\eta)t_k, (c^*-\eta)t_k\big),\\
        (u_1,v_1) &\text{ otherwise.}
    \end{cases}
\end{equation}
     Let 
     $$
     u_k(t,x) \coloneqq u(t + t_k,x+x_k)\quad \text{ and }\quad v_k(t,x) \coloneqq v(t+t_k,x+x_k),$$ 
     then
     $$
     u_k(t,x) \geq \delta_0 \quad \text{ in }\Omega_k,
     $$
where
$$
\Omega_k = \{(t,x): t + t_k \geq T,\quad  -\eta (t+t_k) \leq x+x_k \leq (c^* - \eta/2)(t+t_k)\}.
$$ 
\noindent {\bf Claim.}
{\it $\Omega_k \to \mathbb{R}^2$, i.e. for any compact subset $K \subset \mathbb{R}^2$ there exists $k_1>1$ such that} 
$$
K \subset \Omega_k \quad \text{ for all }k \geq k_1.
$$

Indeed, given $K$, choose $R>0$ such that $K \subset [-R,R]\times [-R,R]$. Then for  $k \gg 1$, we have
\begin{align*}
    K &\subset [-R,R]^2\\ &\subseteq \{(t,x): |t|\leq R,\quad \eta(R - t_k) \leq x \leq   -(c^* - \eta/2)R + \eta t_k/2\}\\
    &\subseteq \{(t,x): |t|\leq R,\quad -\eta(t+t_k) \leq x \leq (c^* - \eta/2) t + \eta t_k /2\}\\
    &\subseteq \{(t,x): t + t_k \geq T,\quad -\eta (t+t_k) - x_k \leq x \leq (c^* - \eta/2) t -x_k + (c^* -\eta/2)t_k\} \\
    &=\Omega_k.
\end{align*}
This proves the claim.

It follows from the claim and Lemma \ref{lemma:gl_bds} that there exists constants $0<\delta_1 <1$ independent of $k$ such that
\begin{equation}\label{e.ss.33}
\delta_1 \leq u_k \leq \frac{1}{\delta_1} \quad \text{ and }\quad \delta_1 \leq v_k \leq \frac{1}{\delta_1} \quad \text{ in }\Omega_k
\end{equation}
Using the above $L^\infty$ bounds and parabolic $L^p$ estimates we may deduce, re-labelling a sub-sequence if necessary, that $(u_k,v_k)$ converges weakly in $W^{2,1,p}_{loc}(\mathbb{R}^2)$ (and strongly in\\ $C^{1+\alpha, (1+\alpha)/2}_{loc}(\mathbb{R}^2)$ thanks to Sobolev embedding) to an entire solution $(u_\infty, v_\infty)$ of the system
    \begin{align*}
        \begin{cases}
        u_t &= u_{xx} +  u(-1-u + \tilde a v) \\
        v_t &= dv_{xx}  + rv(1 - v - bu),
        \end{cases}
    \end{align*}
    where $\tilde a = a_2$ if $c_1 < c^*$ and $x_k \in \big((c_1+\eta)t_k, (c^*-\eta)t_k\big)$, and $\tilde a = a_1$ otherwise. Moreover, \eqref{e.ss.33} also implies that
\begin{equation}\label{e.ss.34}
\delta_1 \leq u_\infty \leq \frac{1}{\delta_1}\quad \text{ and }\quad \delta_1 \leq v_\infty \leq \frac{1}{\delta_1}\quad \text{ in }\mathbb{R}^2.
\end{equation}


   Having established the positive upper and lower bounds for $(u_\infty,v_\infty)$ on $\mathbb{R}^2$, one can then repeat a standard argument via Lyapunov functional (see the proof of Lemma 4.1 in \cite{Ducrot2021asymptotic})  that  $(u_\infty,v_\infty)$ is identically equal to the homogeneous steady state $(\tilde{u},\tilde{v})$ given in \eqref{e.ss.35}, i.e., $(u_k,v_k) \to (\tilde u,\tilde v)$ in $C_{loc}(\mathbb{R}^2)$. This in particular implies
$$
 (u,v)(t_k,x_k) =  (u_k,v_k)(0,0)  \to (\tilde u, \tilde v) \quad \text{ as }k \to \infty.
$$
But this is a contradiction, which completes the proof.
\end{proof}

\appendix
\section{Comparison Principle}
Recall the  Hamilton-Jacobi equation
\begin{align}\label{eq:rho_hj'}
   \min\{\rho - s\rho' + |\rho'|^2 + \tilde{R}(s), \rho\} = 0 \quad \text{for } s \in (0,\infty).
\end{align}
We prove a comparison principle for \eqref{eq:rho_hj'} for discontinuous $\tilde R:\mathbb{R}\to\mathbb{R}$ that is locally monotone \cite{Chen2008viscosity}.
\begin{definition}
    A function $h : \mathbb R \to \mathbb R$ is locally monotone if for every $s_0 \in \mathbb R$, either
    \begin{align*}
        \lim_{\delta \to 0} \inf_{\substack{|s_i-s_0| < \delta \\ s_1 > s_2}} \big(h(s_1) - h(s_2)\big) \geq 0 \quad \text{ or } \quad \lim_{\delta \to 0} \sup_{\substack{|s_i-s_0| < \delta \\ s_1 > s_2}} \big(h(s_1) - h(s_2)\big) \leq 0. 
    \end{align*}
\end{definition} 
The assumptions on $\tilde R$ are stated precisely as follows.
\begin{itemize}
    \item[{(B1)}] $\tilde R(s)$ is locally monotone;
    \item[{(B2)}] $\tilde R^*(s) = \tilde R_*(s)$ almost everywhere, and $\inf_{s>0}R(s) > 0$, where $\tilde R^*$ and $\tilde R_*$ are defined by
\begin{align*}
    \tilde R^*(s) = \limsup_{s' \to s} \tilde R(s') \quad \text{ and } \quad \tilde R_*(s) = \liminf_{s' \to s} \tilde R(s').
\end{align*}
\end{itemize}

\begin{lemma}
\label{lemma:comp}
    Suppose $\tilde R(s)$ satisfies {\rm(B1)-(B2)}. Let $\overline \rho$ and $\underline \rho$ be non-negative viscosity super- and sub-solutions, respectively, of \eqref{eq:rho_hj'} 
 such that 
\begin{align}\label{eq:bd_comp}
        \underline \rho(0) \leq \overline \rho(0) \quad \text{and} \quad \lim_{s \to \infty} \frac{\underline \rho(s)}{s} \leq \lim_{s \to \infty} \frac{\overline \rho(s)}{s}.
    \end{align}
    Then $\underline \rho \leq \overline \rho$ in $(0,\infty)$.
\end{lemma}
\begin{proof}
We apply the results of \cite{lam2022asymptotic}. The specific form of the Hamilton-Jacobi equation \eqref{eq:rho_hj'} and assumptions (B1)-(B2), imply that \cite[{\rm(H1)-(H6)}]{lam2022asymptotic} hold. Hence, \cite[Proposition 2.11]{lam2022asymptotic} applies.
\end{proof}

\begin{corollary}\label{corollary:uniqueness}
Let $\tilde{R}:\mathbb{R}\to\mathbb{R}$ satisfy the assumptions {\rm(B1)-(B2)}. Then there exists a unique viscosity solution
$\hat \rho$ to\eqref{eq:rho_hj'}  satisfying the boundary conditions 
    \begin{align}
    \hat\rho(0) = 0 \quad \text{ and } \lim_{s \to \infty} \frac{\hat \rho(s)}{s} = \infty.
    \end{align}
Moreover, $\hat\rho$ is nondecreasing.
\end{corollary}
\begin{proof}
Thanks to {\rm(B1)-(B2)},  \cite[Proposition 1.7(b)]{lam2022asymptotic} applies. Next, by applying \cite[Lemma 2.9]{lam2022asymptotic}, we deduce that $s\mapsto \hat\rho(s)$ is nondecreasing.
\end{proof}

{\small

\bibliographystyle{siam}

\begin{thebibliography}{10}

\bibitem{ahn2022spreading}
{\sc I.~Ahn, W.~Choi, A.~Ducrot, and J.-S. Guo}, {\em Spreading dynamics for a
  three species predator–prey system with two preys in a shifting
  environment}, J. Dyn. Differ. Equ.,  (2022).

\bibitem{aronson1975nonlinear}
{\sc D.~G. Aronson and H.~F. Weinberger}, {\em Nonlinear diffusion in
  population genetics, combustion, and nerve pulse propagation}, in Partial
  differential equations and related topics ({P}rogram, {T}ulane {U}niv., {N}ew
  {O}rleans, {L}a., 1974), Lecture Notes in Math., Vol. 446, Springer,
  Berlin-New York, 1975, pp.~5--49.

\bibitem{aronson1978multidimensional}
\leavevmode\vrule height 2pt depth -1.6pt width 23pt, {\em Multidimensional
  nonlinear diffusion arising in population genetics}, Adv. in Math., 30
  (1978), pp.~33--76.

\bibitem{aronson2015no}
{\sc R.~B. Aronson, K.~E. Smith, S.~C. Vos, J.~B. McClintock, M.~O. Amsler,
  P.-O. Moksnes, D.~S. Ellis, J.~Kaeli, H.~Singh, J.~W. Bailey, et~al.}, {\em
  No barrier to emergence of bathyal king crabs on the antarctic shelf},
  Proceedings of the National Academy of Sciences, 112 (2015),
  pp.~12997--13002.

\bibitem{Bardi1997optimal}
{\sc M.~Bardi and I.~Capuzzo-Dolcetta}, {\em Optimal control and viscosity
  solutions of {H}amilton-{J}acobi-{B}ellman equations}, Systems \& Control:
  Foundations \& Applications, Birkh\"{a}user Boston, Inc., Boston, MA, 1997.
\newblock With appendices by Maurizio Falcone and Pierpaolo Soravia.

\bibitem{Barles2013introduction}
{\sc G.~Barles}, {\em An introduction to the theory of viscosity solutions for
  first-order {H}amilton-{J}acobi equations and applications}, in
  Hamilton-{J}acobi equations: approximations, numerical analysis and
  applications, vol.~2074 of Lecture Notes in Math., Springer, Heidelberg,
  2013, pp.~49--109.

\bibitem{Barles1987discontinuous}
{\sc G.~Barles and B.~Perthame}, {\em Discontinuous solutions of deterministic
  optimal stopping time problems}, RAIRO Mod\'{e}l. Math. Anal. Num\'{e}r., 21
  (1987), pp.~557--579.

\bibitem{berestycki2009can}
{\sc H.~Berestycki, O.~Diekmann, C.~J. Nagelkerke, and P.~A. Zegeling}, {\em
  Can a species keep pace with a shifting climate?}, Bull. Math. Biol., 71
  (2009), pp.~399--429.

\bibitem{berestycki2018forced}
{\sc H.~Berestycki and J.~Fang}, {\em Forced waves of the {F}isher-{KPP}
  equation in a shifting environment}, J. Differential Equations, 264 (2018),
  pp.~2157--2183.

\bibitem{Berestycki2016shape}
{\sc H.~Berestycki, J.-M. Roquejoffre, and L.~Rossi}, {\em The shape of
  expansion induced by a line with fast diffusion in {F}isher-{KPP} equations},
  Comm. Math. Phys., 343 (2016), pp.~207--232.

\bibitem{bradley2010climate}
{\sc B.~A. Bradley, D.~S. Wilcove, and M.~Oppenheimer}, {\em Climate change
  increases risk of plant invasion in the eastern united states}, Biological
  Invasions, 12 (2010), pp.~1855--1872.

\bibitem{bramson1983convergence}
{\sc M.~Bramson}, {\em Convergence of solutions of the {K}olmogorov equation to
  travelling waves}, Mem. Amer. Math. Soc., 44 (1983), pp.~iv+190.

\bibitem{buckley2012functional}
{\sc L.~B. Buckley and J.~G. Kingsolver}, {\em Functional and phylogenetic
  approaches to forecasting species' responses to climate change}, Annual
  Review of Ecology, Evolution, and Systematics, 43 (2012), pp.~205--226.

\bibitem{Chen2008viscosity}
{\sc X.~Chen and B.~Hu}, {\em Viscosity solutions of discontinuous
  {H}amilton-{J}acobi equations}, Interfaces Free Bound., 10 (2008),
  pp.~339--359.

\bibitem{choi2021persistence}
{\sc W.~Choi, T.~Giletti, and J.-S. Guo}, {\em Persistence of species in a
  predator-prey system with climate change and either nonlocal or local
  dispersal}, J. Differential Equations, 302 (2021), pp.~807--853.

\bibitem{daugaard2019}
{\sc U.~Daugaard, O.~L. Petchey, and F.~Pennekamp}, {\em Warming can
  destabilize predator-prey interactions by shifting the functional response
  from type iii to type ii}, J. Anim. Ecol., 88 (2019), pp.~1575--1586.

\bibitem{deutsch2008impacts}
{\sc C.~A. Deutsch, J.~J. Tewksbury, R.~B. Huey, K.~S. Sheldon, C.~K.
  Ghalambor, D.~C. Haak, and P.~R. Martin}, {\em Impacts of climate warming on
  terrestrial ectotherms across latitude}, Proceedings of the National Academy
  of Sciences, 105 (2008), pp.~6668--6672.

\bibitem{Dong2021forced}
{\sc F.-D. Dong, B.~Li, and W.-T. Li}, {\em Forced waves in a
  {L}otka-{V}olterra competition-diffusion model with a shifting habitat}, J.
  Differential Equations, 276 (2021), pp.~433--459.

\bibitem{ducrot2013convergence}
{\sc A.~Ducrot}, {\em Convergence to generalized transition waves for some
  {H}olling-{T}anner prey-predator reaction-diffusion system}, J. Math. Pures
  Appl. (9), 100 (2013), pp.~1--15.

\bibitem{ducrot2016spatial}
\leavevmode\vrule height 2pt depth -1.6pt width 23pt, {\em Spatial propagation
  for a two component reaction-diffusion system arising in population
  dynamics}, J. Differential Equations, 260 (2016), pp.~8316--8357.

\bibitem{Ducrot2021asymptotic}
{\sc A.~Ducrot, T.~Giletti, J.-S. Guo, and M.~Shimojo}, {\em Asymptotic
  spreading speeds for a predator-prey system with two predators and one prey},
  Nonlinearity, 34 (2021), pp.~669--704.

\bibitem{ducrot2019spreading}
{\sc A.~Ducrot, T.~Giletti, and H.~Matano}, {\em Spreading speeds for
  multidimensional reaction-diffusion systems of the prey-predator type}, Calc.
  Var. Partial Differential Equations, 58 (2019), pp.~Paper No. 137, 34.

\bibitem{ducrot2019thespreading}
{\sc A.~Ducrot, J.-S. Guo, G.~Lin, and S.~Pan}, {\em The spreading speed and
  the minimal wave speed of a predator-prey system with nonlocal dispersal}, Z.
  Angew. Math. Phys., 70 (2019), pp.~Paper No. 146, 25.

\bibitem{dunbar1983travelling}
{\sc S.~R. Dunbar}, {\em Travelling wave solutions of diffusive
  {L}otka-{V}olterra equations}, J. Math. Biol., 17 (1983), pp.~11--32.

\bibitem{Evans1989PDE}
{\sc L.~C. Evans and P.~E. Souganidis}, {\em A {PDE} approach to geometric
  optics for certain semilinear parabolic equations}, Indiana Univ. Math. J.,
  38 (1989), pp.~141--172.

\bibitem{fang2016can}
{\sc J.~Fang, Y.~Lou, and J.~Wu}, {\em Can pathogen spread keep pace with its
  host invasion?}, SIAM J. Appl. Math., 76 (2016), pp.~1633--1657.

\bibitem{fisher1937wave}
{\sc R.~A. Fisher}, {\em The wave of advance of advantageous genes}, Annals of
  eugenics, 7 (1937), pp.~355--369.

\bibitem{gardner1991stability}
{\sc R.~Gardner and C.~K. R.~T. Jones}, {\em Stability of travelling wave
  solutions of diffusive predator-prey systems}, Trans. Amer. Math. Soc., 327
  (1991), pp.~465--524.

\bibitem{gardner1984existence}
{\sc R.~A. Gardner}, {\em Existence of travelling wave solutions of
  predator-prey systems via the connection index}, SIAM J. Appl. Math., 44
  (1984), pp.~56--79.

\bibitem{gilman2010framework}
{\sc S.~E. Gilman, M.~C. Urban, J.~Tewksbury, G.~W. Gilchrist, and R.~D. Holt},
  {\em A framework for community interactions under climate change}, Trends in
  ecology \& evolution, 25 (2010), pp.~325--331.

\bibitem{girardin2019invasion}
{\sc L.~Girardin and K.-Y. Lam}, {\em Invasion of open space by two
  competitors: spreading properties of monostable two-species
  competition-diffusion systems}, Proc. Lond. Math. Soc. (3), 119 (2019),
  pp.~1279--1335.

\bibitem{guo2023spreading}
{\sc J.-S. Guo, M.~Shimojo, and C.-C. Wu}, {\em Spreading dynamics for a
  predator-prey system with two predators and one prey in a shifting habitat},
  Discrete Contin. Dyn. Syst. Ser. B, 28 (2023), pp.~6126--6141.

\bibitem{Hamel2012spreading}
{\sc F.~Hamel and G.~Nadin}, {\em Spreading properties and complex dynamics for
  monostable reaction-diffusion equations}, Comm. Partial Differential
  Equations, 37 (2012), pp.~511--537.

\bibitem{hobbs2000invasive}
{\sc R.~J. Hobbs and H.~A. Mooney}, {\em Invasive species in a changing world},
  Island press, 2000.

\bibitem{holzer2014accelerated}
{\sc M.~Holzer and A.~Scheel}, {\em Accelerated fronts in a two-stage invasion
  process}, SIAM J. Math. Anal., 46 (2014), pp.~397--427.

\bibitem{huang2003existence}
{\sc J.~Huang, G.~Lu, and S.~Ruan}, {\em Existence of traveling wave solutions
  in a diffusive predator-prey model}, J. Math. Biol., 46 (2003), pp.~132--152.

\bibitem{huey1979integrating}
{\sc R.~B. Huey and R.~Stevenson}, {\em Integrating thermal physiology and
  ecology of ectotherms: a discussion of approaches}, American Zoologist, 19
  (1979), pp.~357--366.

\bibitem{Ishii1985hamilton}
{\sc H.~Ishii}, {\em Hamilton-{J}acobi equations with discontinuous
  {H}amiltonians on arbitrary open sets}, Bull. Fac. Sci. Engrg. Chuo Univ., 28
  (1985), pp.~33--77.

\bibitem{kolmogorov1937etude}
{\sc A.~Kolmogorov, I.~Petrovskii, and N.~Piscounov}, {\em \'{E}tude de
  l’\'{e}quation de la diffusion avec croissance de la quantit \'{e} de mati
  \'{e}re et son application a un probl\'{e}me biologique}, Moscow Univ. Bull.
  Ser. Boarding school. Section. HAS, 1 (1937), p.~126.

\bibitem{lam2022introduction}
{\sc K.-Y. Lam and Y.~Lou}, {\em Introduction to Reaction-Diffusion Equations:
  Theory and Applications to Spatial Ecology and Evolutionary Biology}, Lecture
  Notes on Mathematical Modelling in the Life Sciences, Springer, Cham, 2022.

\bibitem{lam2022asymptotic}
{\sc K.-Y. Lam and X.~Yu}, {\em Asymptotic spreading of {KPP} reactive fronts
  in heterogeneous shifting environments}, J. Math. Pures Appl. (9), 167
  (2022), pp.~1--47.

\bibitem{lang2017temperature}
{\sc B.~Lang, R.~B. Ehnes, U.~Brose, and B.~C. Rall}, {\em Temperature and
  consumer type dependencies of energy flows in natural communities}, Oikos,
  126 (2017), p.~1717–1725.

\bibitem{lau1985on}
{\sc K.-S. Lau}, {\em On the nonlinear diffusion equation of {K}olmogorov,
  {P}etrovsky, and {P}iscounov}, J. Differential Equations, 59 (1985),
  pp.~44--70.

\bibitem{lenoir2020species}
{\sc J.~Lenoir, R.~Bertrand, L.~Comte, L.~Bourgeaud, T.~Hattab, J.~Murienne,
  and G.~Grenouillet}, {\em Species better track climate warming in the oceans
  than on land}, Nat. Ecol. Evol., 4 (2020), pp.~1044--1059.

\bibitem{lewis2002spreading}
{\sc M.~A. Lewis, B.~Li, and H.~F. Weinberger}, {\em Spreading speed and linear
  determinacy for two-species competition models}, J. Math. Biol., 45 (2002),
  pp.~219--233.

\bibitem{li2014persistence}
{\sc B.~Li, S.~Bewick, J.~Shang, and W.~F. Fagan}, {\em Persistence and spread
  of a species with a shifting habitat edge}, SIAM J. Appl. Math., 74 (2014),
  pp.~1397--1417.

\bibitem{li2005spreading}
{\sc B.~Li, H.~F. Weinberger, and M.~A. Lewis}, {\em Spreading speeds as
  slowest wave speeds for cooperative systems}, Math. Biosci., 196 (2005),
  pp.~82--98.

\bibitem{liang2007asymptotic}
{\sc X.~Liang and X.-Q. Zhao}, {\em Asymptotic speeds of spread and traveling
  waves for monotone semiflows with applications}, Comm. Pure Appl. Math., 60
  (2007), pp.~1--40.

\bibitem{ling2008range}
{\sc S.~D. Ling}, {\em Range expansion of a habitat-modifying species leads to
  loss of taxonomic diversity: a new and impoverished reef state}, Oecologia,
  156 (2008), pp.~883--894.

\bibitem{Liu2020Asymptotic}
{\sc Q.~Liu, S.~Liu, and K.-Y. Lam}, {\em Asymptotic spreading of interacting
  species with multiple fronts {I}: a geometric optics approach}, Discrete
  Contin. Dyn. Syst., 40 (2020), pp.~3683--3714.

\bibitem{Liu2021stacked}
\leavevmode\vrule height 2pt depth -1.6pt width 23pt, {\em Stacked invasion
  waves in a competition-diffusion model with three species}, J. Differential
  Equations, 271 (2021), pp.~665--718.

\bibitem{mischaikow1993travelling}
{\sc K.~Mischaikow and J.~F. Reineck}, {\em Travelling waves in predator-prey
  systems}, SIAM J. Math. Anal., 24 (1993), pp.~1179--1214.

\bibitem{pan2013asymptotic}
{\sc S.~Pan}, {\em Asymptotic spreading in a {L}otka-{V}olterra predator-prey
  system}, J. Math. Anal. Appl., 407 (2013), pp.~230--236.

\bibitem{pan2017invasion}
\leavevmode\vrule height 2pt depth -1.6pt width 23pt, {\em Invasion speed of a
  predator-prey system}, Appl. Math. Lett., 74 (2017), pp.~46--51.

\bibitem{parmesan2003}
{\sc C.~Parmesan and G.~Yohe}, {\em A globally coherent fingerprint of climate
  change impacts across natural systems}, Nature, 421 (2003), pp.~37--42.

\bibitem{pecl2017}
{\sc G.~T. Pecl, M.~B. Araújo, J.~D. Bell, J.~Blanchard, T.~C. Bonebrake,
  I.-C. Chen, T.~D. Clark, R.~K. Colwell, F.~Danielsen, B.~Evengård,
  L.~Falconi, S.~Ferrier, S.~Frusher, R.~A. Garcia, R.~B. Griffis, A.~J.
  Hobday, C.~Janion-Scheepers, M.~A. Jarzyna, S.~Jennings, J.~Lenoir, H.~I.
  Linnetved, V.~Y. Martin, P.~C. McCormack, J.~McDonald, N.~J. Mitchell,
  T.~Mustonen, J.~M. Pandolfi, N.~Pettorelli, E.~Popova, S.~A. Robinson, B.~R.
  Scheffers, J.~D. Shaw, C.~J.~B. Sorte, J.~M. Strugnell, J.~M. Sunday, M.-N.
  Tuanmu, A.~Vergés, C.~Villanueva, T.~Wernberg, E.~Wapstra, and S.~E.
  Williams}, {\em Biodiversity redistribution under climate change: Impacts on
  ecosystems and human well-being}, Science, 355 (2017).

\bibitem{potapov2004climate}
{\sc A.~B. Potapov and M.~A. Lewis}, {\em Climate and competition: the effect
  of moving range boundaries on habitat invasibility}, Bull. Math. Biol., 66
  (2004), pp.~975--1008.

\bibitem{Protter1984maximum}
{\sc M.~H. Protter and H.~F. Weinberger}, {\em Maximum principles in
  differential equations}, Springer-Verlag, New York, 1984.
\newblock Corrected reprint of the 1967 original.

\bibitem{shigesada1997biological}
{\sc N.~Shigesada and K.~Kawasaki}, {\em Biological Invasions: Theory and
  Practice}, Oxford Series in Ecology and Evolution, Oxford University Press
  Inc., New York, 1997.

\bibitem{Smoller1983shock}
{\sc J.~Smoller}, {\em Shock waves and reaction-diffusion equations}, vol.~258
  of Grundlehren der Mathematischen Wissenschaften, Springer-Verlag, New
  York-Berlin, 1983.

\bibitem{sorte2010}
{\sc C.~J.~B. Sorte, S.~L. Williams, and J.~T. Carlton}, {\em Marine range
  shifts and species introductions: Comparative spread rates and community
  impacts}, Glob. Ecol. Biogeogr., 19 (2010), pp.~303--316.

\bibitem{verges2016long}
{\sc A.~Verg{\'e}s, C.~Doropoulos, H.~A. Malcolm, M.~Skye, M.~Garcia-Piz{\'a},
  E.~M. Marzinelli, A.~H. Campbell, E.~Ballesteros, A.~S. Hoey,
  A.~Vila-Concejo, et~al.}, {\em Long-term empirical evidence of ocean warming
  leading to tropicalization of fish communities, increased herbivory, and loss
  of kelp}, Proceedings of the National Academy of Sciences, 113 (2016),
  pp.~13791--13796.

\bibitem{wallingford2020adjusting}
{\sc P.~D. Wallingford, T.~L. Morelli, J.~M. Allen, E.~M. Beaury, D.~M.
  Blumenthal, B.~A. Bradley, J.~S. Dukes, R.~Early, E.~J. Fusco, D.~E.
  Goldberg, et~al.}, {\em Adjusting the lens of invasion biology to focus on
  the impacts of climate-driven range shifts}, Nature Climate Change, 10
  (2020), pp.~398--405.

\bibitem{walther2009alien}
{\sc G.-R. Walther, A.~Roques, P.~E. Hulme, M.~T. Sykes, P.~Py{\v{s}}ek,
  I.~K{\"u}hn, M.~Zobel, S.~Bacher, Z.~Botta-Duk{\'a}t, H.~Bugmann, et~al.},
  {\em Alien species in a warmer world: risks and opportunities}, Trends in
  ecology \& evolution, 24 (2009), pp.~686--693.

\bibitem{Wang2011spreading}
{\sc H.~Wang}, {\em Spreading speeds and traveling waves for non-cooperative
  reaction-diffusion systems}, J. Nonlinear Sci., 21 (2011), pp.~747--783.

\bibitem{wang2022recent}
{\sc J.-B. Wang, W.-T. Li, F.-D. Dong, and S.-X. Qiao}, {\em Recent
  developments on spatial propagation for diffusion equations in shifting
  environments}, Discrete Contin. Dyn. Syst. Ser. B, 27 (2022), pp.~5101--5127.

\bibitem{weinberger2002analysis}
{\sc H.~F. Weinberger, M.~A. Lewis, and B.~Li}, {\em Analysis of linear
  determinacy for spread in cooperative models}, J. Math. Biol., 45 (2002),
  pp.~183--218.

\bibitem{weiskopf2020climate}
{\sc S.~R. Weiskopf, M.~A. Rubenstein, L.~G. Crozier, S.~Gaichas, R.~Griffis,
  J.~E. Halofsky, K.~J. Hyde, T.~L. Morelli, J.~T. Morisette, R.~C. Mu{\~n}oz,
  et~al.}, {\em Climate change effects on biodiversity, ecosystems, ecosystem
  services, and natural resource management in the united states}, Science of
  the Total Environment, 733 (2020), p.~137782.

\bibitem{wu2019spreading}
{\sc C.-C. Wu}, {\em The spreading speed for a predator-prey model with one
  predator and two preys}, Appl. Math. Lett., 91 (2019), pp.~9--14.

\bibitem{Wu2023propagation}
{\sc S.-L. Wu, L.~Pang, and S.~Ruan}, {\em Propagation dynamics in periodic
  predator-prey systems with nonlocal dispersal}, J. Math. Pures Appl. (9), 170
  (2023), pp.~57--95.

\bibitem{wu2016impact}
{\sc X.~Wu, Y.~Lu, S.~Zhou, L.~Chen, and B.~Xu}, {\em Impact of climate change
  on human infectious diseases: Empirical evidence and human adaptation},
  Environment international, 86 (2016), pp.~14--23.

\bibitem{Li2021invasion}
{\sc F.~Yang, W.~Li, and R.~Wang}, {\em Invasion waves for a nonlocal dispersal
  predator-prey model with two predators and one prey}, Commun. Pure Appl.
  Anal., 20 (2021), pp.~4083--4105.

\bibitem{Zhang2017persistence}
{\sc Z.~Zhang, W.~Wang, and J.~Yang}, {\em Persistence versus extinction for
  two competing species under a climate change}, Nonlinear Anal. Model.
  Control, 22 (2017), pp.~285--302.

\end{thebibliography}

}

\end{document}